%% file: Torus_Graphs.tex
\documentclass{article}
\usepackage{arxiv}
\RequirePackage{natbib}
\RequirePackage{amsthm,amsmath}
\usepackage[utf8]{inputenc} 
\usepackage[T1]{fontenc}    
\usepackage{hyperref}       
\usepackage{url}            
\usepackage{booktabs}       
\usepackage{amsfonts}       
\usepackage{nicefrac}       
\usepackage{microtype}      
\usepackage{lipsum} 

\numberwithin{equation}{section}
\theoremstyle{plain}
\newtheorem{thm}{Theorem}[section]
\newtheorem{corollary}{Corollary}[thm]

\newcommand{\trial}[1]{#1^{(n)}}
\newcommand{\mbf}[1]{\mathbf{#1}}
\newcommand{\bs}[1]{\boldsymbol{#1}}
\newtheorem{lemma}{Lemma}[section]

\newcommand{\sgn}{\text{\textnormal{sgn}}}
\usepackage{amssymb}
\usepackage{enumitem}
\usepackage{url}
\usepackage{bbm}
\usepackage{graphicx}

\title{Torus Graphs for Multivariate Phase Coupling Analysis 
\thanks{This is the peer reviewed version of the following
    article: Natalie Klein, Josue Orellana, Scott Brincat, Earl K. Miller, and Robert E. Kass, ``Torus Graphs for Multivariate Phase Coupling Analysis,'' \textit{in press at The Annals of Applied Statistics}. This article may be used for non-commercial purposes in accordance with IMS Journals and Conditions for Use of Self-Archived Versions.}}

\author{Natalie Klein$^\bot{}^\S$
\thanks{N.K. and J.O. contributed equally to this work.}
, Josue Orellana$^\P{}^\S{}^\dagger$, Scott Brincat$^\parallel$, Earl K. Miller$^\parallel$, and Robert E. Kass$^\bot{}^\P{}^\S$
\thanks{To whom correspondence should be addressed. E-mail: \texttt{kass@stat.cmu.edu}} \\ 
Department of Statistics and Data Science, Carnegie Mellon University$^\bot$,\\
Machine Learning Department, Carnegie Mellon University$^\S$, \\ 
Center for the Neural Basis of Cognition, Carnegie Mellon University and University of Pittsburgh$^\P$\\
and Department of Brain and Cognitive Science, Massachusetts Institute of Technology$^\parallel$ 
}

\begin{document}
\maketitle
\input{tex/abstract.tex}
\section{Introduction}
\input{tex/introduction.tex}
\section{Torus graph model}
\label{sec:tgmodel}
\input{tex/tgmodel.tex}

\section{Important subfamilies of the torus graph model}
\label{sec:tgsubmodels}
\input{tex/tgsubmodels.tex}

\section{Phase coupling in torus graphs}
\label{sec:tgphasecoupling}
\input{tex/phasecouplingintg.tex}
\section{Torus graph estimation and inference}
\label{sec:tgest}
\input{tex/tgestimation.tex}

\section{Simulation study}
\label{sec:simstudy}
\input{tex/simstudy.tex}

\section{Analysis of neural phase angles}
\label{sec:lfpdata}
\input{tex/lfpdata.tex}
\section{Discussion}
\label{sec:disc}
\input{tex/discussion.tex}

\section*{Acknowledgements}
The authors thank the reviewers for their valuable feedback. 
REK, NK, and JO are supported by National Institute of Mental Health (NIMH, \url{http://www.nimh.nih.gov}, R01MH064537). 
JO is also partially supported by the Commonwealth of Pennsylvania Department of Health SAP4100072542, and Presidential Ph.D. Fellowships from The Richard King Mellon Foundation and William S. Dietrich II.
EKM and SB are supported by National Institute of Mental Health (NIMH, \url{http://www.nimh.nih.gov}, R37MH087027) and The MIT Picower Institute Innovation Fund. 
The funders had no role in study design, data collection and analysis, decision to publish, or preparation of the manuscript.

\section*{Supplementary code and tutorial}
\url{https://github.com/natalieklein/torus-graphs}

\section*{Supplementary text and figures}
\renewcommand{\thefigure}{S\arabic{figure}}
\renewcommand{\theHfigure}{S\arabic{figure}}
\setcounter{figure}{0} 
\renewcommand{\theequation}{S\arabic{equation}}
\renewcommand{\thesubsection}{S\arabic{subsection}}
\renewcommand{\thethm}{S\arabic{subsection}.1}
Several sections of supplemental text with derivations, proofs, and further detailed discussion of points in the main text. Also included in this section are extra supporting figures (referenced in the main text) with descriptive captions to assist in interpretation.

\subsection{Proof of Theorem \ref{thm:tgnatural} (Torus graph model)} 
\label{sec:tgmodelcomplex}
\input{texsupp/tgmodelcomplex.tex}

\subsection{Reparameterization to compare to previous work}
\label{sec:s2}
\input{texsupp/tgreparam.tex}

\subsection{Proof of Corollary \ref{cor:tgprop} (Torus graph properties)}
\label{sec:s3}
\input{texsupp/tgtheoremproof.tex}

\subsection{Derivations of phase differences in torus graph models}
\label{sec:s4}
\input{texsupp/phasediff.tex}

\subsection{Proof of Theorem \ref{thm:scorematch} (Score matching estimators for torus graphs)}
\label{sec:s5}
\input{texsupp/scorematch.tex}

\subsection{Proof of Theorem \ref{thm:conditional} (Conditional distributions in torus graphs)}
\label{sec:s6}
\input{texsupp/conditional.tex}

\subsection{Measures of positive and negative circular dependence}
\label{sec:s7}
\input{texsupp/posnegdepend.tex}
\input{texsupp/figures.tex}

\clearpage
\bibliographystyle{apa}
\bibliography{Torus_Graphs}

\end{document}

%% file: tex/abstract.tex
\begin{abstract}
Angular measurements are often modeled as circular random variables, where there are natural circular analogues of moments, including correlation. 
Because a product of circles is a torus, a $d$-dimensional vector of circular random variables lies on a $d$-dimensional torus. 
For such vectors we present here a class of graphical models, which we call {\em torus graphs}, based on the full exponential family with pairwise interactions. 
The topological distinction between a torus and Euclidean space has several important consequences. 

Our development was motivated by the problem of identifying phase coupling among oscillatory signals recorded from multiple electrodes in the brain: oscillatory phases across electrodes might tend to advance or recede together, indicating coordination across brain areas. 
The data analyzed here consisted of 24 phase angles measured repeatedly across 840 experimental trials (replications) during a memory task, where the electrodes were in 4 distinct brain regions, all known to be active while memories are being stored or retrieved. 
In realistic numerical simulations, we found that a standard pairwise assessment, known as phase locking value, is unable to describe multivariate phase interactions, but that torus graphs can accurately identify conditional associations. 
Torus graphs generalize several more restrictive approaches that have appeared in various scientific literatures, and produced intuitive results in the data we analyzed. 
Torus graphs thus unify multivariate analysis of circular data and present fertile territory for future research.
\end{abstract}

%% file: tex/introduction.tex
New technologies for recording electrical activity among large networks of neurons have created great opportunities to advance neurophysiology, and great challenges in data analysis \citep[e.g.,][]{steinmetz2018challenges}. 
One appealing idea, which has garnered substantial attention, is that under certain circumstances, long-range communication across brain areas may be facilitated through coordinated network oscillations \citep{buzsaki2004neuronal,ching2010thalamocortical,fell2011role,sherman2016neural}. 
To demonstrate coordination among oscillatory networks, computational neuroscientists have examined \textit{phase coupling} across replications (trials) of the experiment. 
That is, when the phase of an oscillatory potential at a particular location, and a particular latency from the beginning of the trial, is measured repeatedly it will vary; phase coupling refers to the tendency of two phases, measured simultaneously at two locations, to vary together, i.e., to be associated, across trials. 
The data we analyze here consist of 24 phase angles recorded simultaneously, on each of many trials, 
from several brain regions known to play a role in memory formation and recall \citep{brincat2015frequency, Brincat:2016cl},  prefrontal cortex (PFC) and three sub-areas of the hippocampus, the dentate gyrus (DG), subiculum (Sub) and CA3. 
Being angles, phases may be considered circular random variables. 
A commonly-applied measure of phase coupling, known as Phase Locking Value (PLV), is an estimator of the natural circular analogue of Pearson correlation under certain assumptions (which we review). 
PLV, however, like correlation, can not distinguish between direct association and indirect association via alternative pathways. 
Thus, a large PLV between PFC and DG does not distinguish between direct coupling and indirect coupling via a third area, such as via CA3 (neural activity in PFC could be coupled directly with that in CA3, and that in CA3 with that in DG). 
To draw such a distinction we need, instead, a circular measure that is analogous to partial correlation. 
More generally, we wish to construct circular analogues of Gaussian graphical models. 
Because key properties of Gaussian graphical models are inherited by exponential families, and the product of circles is a torus, we consider exponential families on a multidimensional torus and call the resulting models \textit{torus graphs}. 
We used torus graphs to provide a thorough description of associations among the 24 repeatedly-measured phases in the Brincat and Miller data and, in particular, we found strong evidence that the association between activity in PFC and DG is indirect, via both CA3 and Sub, rather than direct.

When circular random variables are highly concentrated around a central value, there is little harm in ignoring their circular nature, and multivariate Gaussian methods could be applied. 
However, in most of the neurophysiological data we have seen, including those analyzed here, the marginal distributions of phases are very diffuse, close to uniform, so the topological distinction between the circle and the real line is important. 
The torus topology is consequential not only for computation of probabilities but also for the interpretation of association. 
Figure \ref{FigRepeatedTrials_part1} displays the inability of rectangular coordinates to preserve the clustering of points around a diagonal line under strong positive association. 
Furthermore, unlike the Gaussian case where a single scalar, correlation, describes both positive and negative association, on the torus, positive and negative association each have both an amplitude and a phase, so each pairwise association is, in general, described by 2 complex numbers. 
Also, in the Gaussian case, it is possible to interpret the association of two variables without knowing their marginal concentrations. 
This is no longer true for torus graphs. 

After defining torus graphs and providing a few basic properties in Section \ref{sec:tgmodel}, in Section \ref{sec:tgsubmodels} we consider several important alternative families of distributions for multivariate circular data that have appeared in the literature, and show that they are all special cases of torus graphs. 
In Section \ref{sec:tgphasecoupling}, we step through the interpretation of phase coupling in torus graphs by considering in detail the bivariate and trivariate cases. 
In Section \ref{sec:tgest}, we provide estimation and inference procedures and, in Section \ref{sec:simstudy}, document via simulation studies the very good performance of these procedures in realistic settings. 
Our analysis of the data appears in Section \ref{sec:lfpdata} and we make a few closing remarks in Section \ref{sec:disc}.

\begin{figure}
    \centering
    \includegraphics[width=0.7\textwidth]{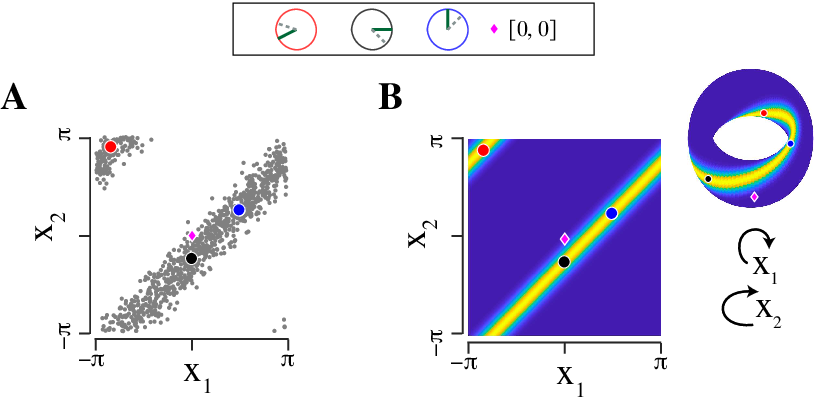}
    \caption{Rectangular coordinates are unable to accurately represent strong positive association between two circular random variables. (A) Scatter plot of simulated observations from a pair of dependent circular variables in rectangular coordinates, with three observations highlighted in red, black, and blue (simulated plot is similar to real data plots, but somewhat more concentrated for visual clarity; see Figure \ref{fig:bivscatterDG}). 
    The highlighted observations are shown on circles at the top of the figure (one angle as a dashed line, one as a solid line; positive dependence implies a consistent offset between the angles).
    While the black and blue points follow the diagonal line, the red point falls near the upper left corner due to conversion to rectangular coordinates. (B) Probability density representing the two variables, plotted both in rectangular coordinates and on the torus, with the same three points marked. 
    On the torus, there is a single band with high probability, which wraps around and connects to itself as a M\"obius strip, and all three points fall on this strip.}
    \label{FigRepeatedTrials_part1}
\end{figure}

%% file: tex/tgmodel.tex
Suppose $\mbf{X}$ is a $d$-dimensional random vector with $j$th element $X_j$ being a circular random variable, which may be expressed as an angle in $[0, 2\pi)$, though other choices of angular intervals, such as $[-\pi, \pi)$, are equivalent.
When $d = 2$, $\mbf{X}$ lies on the product of two circles (a torus), and in general it lies on a multidimensional torus.
When considering phase coupling in neural data, $\mbf{X}$ represents a vector of phase angle values extracted from oscillatory signals for a single time point from each of $d$ signals, with repeated trials providing multiple observations. 

The torus graph model may be developed by analogy to the multivariate Gaussian distribution, a member of the exponential family that models dependence between $d$ real-valued variables.
In general, for a random vector $\mbf{Y}$, an exponential family distribution is specified through a vector of natural parameters $\bs{\eta}$ that multiply a vector of sufficient statistics $\mbf{S}(\mbf{y})$ summarizing information from the data that is sufficient for the parameters
\citep{wainwright2008graphical} and has a density of the form:
\begin{align*}
    p(\mbf{y} ; \bs{\eta}) \propto \exp \left( \bs{\eta}^T \mbf{S}(\mbf{y}) \right).
\end{align*}

In the bivariate Gaussian distribution, $\mbf{Y} \in \mathbb{R}^2$ and the sufficient statistics corresponding to the natural parameters are $\mbf{y}$ and $\mbf{y}\mbf{y}^T$, which describe the first- and second-order behavior of the variates.
For a vector of angular variables $\mbf{X} \in [0, 2\pi)^2$, we follow \cite{Mardia:2005dl} by representing the angles using rectangular coordinates on the unit circle as $\mbf{Y}_1 = [\cos X_1, \sin X_1]$ and $\mbf{Y}_2 = [\cos X_2, \sin X_2]$.
The first-order sufficient statistics are $\mbf{y}_1$ and $\mbf{y}_2$.
The second-order behavior is described by: 
\begin{align*}
    \mbf{y}_1 \mbf{y}_2^T = \begin{bmatrix} \cos x_1 \cos x_2 & \cos x_1 \sin x_2 \\ \sin x_1 \cos x_2 & \sin x_1 \sin x_2 \end{bmatrix}.
\end{align*}
This choice of sufficient statistics leads to the following natural exponential family density parameterized by $\bs{\eta} = [\bs{\eta}_1, \bs{\eta}_2, \bs{\eta}_{12}]$:
\begin{align}
    p(\mbf{x} ;  \bs{\eta}) \propto \exp \left( \bs{\eta}_1^T 
    \begin{bmatrix} \cos x_1 \\ \sin x_1 \end{bmatrix} +
    \bs{\eta}_2^T 
    \begin{bmatrix} \cos x_2 \\ \sin x_2 \end{bmatrix} +
    \bs{\eta}_{12}^T 
    \begin{bmatrix} \cos x_1 \cos x_2 \\ \cos x_1 \sin x_2 \\ \sin x_1 \cos x_2 \\ \sin x_1 \sin x_2 \end{bmatrix} 
    \right) \label{eq:bivprodmodel}
\end{align}
The first two terms correspond to marginal circular means and concentrations of each variable, while the third term is a pairwise coupling term describing dependence between the variables.
In the absence of pairwise coupling, the marginal distributions are all von Mises \citep[p. 48]{NIFisher1993}, and if $d = 1$, the torus graph model is itself von Mises. 
Extending Equation \ref{eq:bivprodmodel} to $d > 2$ yields
\begin{align}
    p(\mbf{x} ;  \bs{\eta}) \propto \exp \left( \sum_{j=1}^d \bs{\eta}_j^T 
    \begin{bmatrix} \cos x_j \\ \sin x_j \end{bmatrix} +
    \sum_{j < k} \bs{\eta}_{jk}^T 
    \begin{bmatrix} \cos x_j \cos x_k \\ \cos x_j \sin x_k \\ \sin x_j \cos x_k \\ \sin x_j \sin x_k \end{bmatrix} 
    \right). \label{eq:ddimprodmodel}
\end{align}
The normalization constant is intractable, though numerical approximations may be used in the bivariate case \citep{Kurz:2015tr}.

Applying trigonometric product-to-sum formulas to the pairwise coupling terms of Equation \ref{eq:ddimprodmodel} yields an equivalent, alternative parameterization in terms of natural parameters $\bs{\phi}$:
\begin{align}
    p(\mbf{x} ;  \bs{\phi}) \propto \exp \left( \sum_{j=1}^d \bs{\phi}_j^T 
    \begin{bmatrix} \cos x_j \\ \sin x_j \end{bmatrix} +
    \sum_{j < k} \bs{\phi}_{jk}^T 
    \begin{bmatrix} \cos(x_j-x_k) \\ \sin(x_j-x_k) \\ \cos(x_j+x_k) \\ \sin(x_j+x_k) \end{bmatrix} 
    \right). \label{eq:tgdiff}
\end{align}
We define a \textit{$d$-dimensional torus graph} to be any member of the family of distributions specified by Equations \ref{eq:ddimprodmodel} or  \ref{eq:tgdiff}.
In the form of Equation \ref{eq:tgdiff}, the sufficient statistics involving only a single angle are
\begin{align*}
    \mbf{S}^1_j(\mbf{x}) &= [\cos (x_j), \, \sin (x_j)]^T,
\end{align*}
and the sufficient statistics involving pairs of angles are
\begin{align*}
    \mbf{S}^2_{jk}(\mbf{x}) &= [\cos (x_j - x_k), \, 
                                \sin (x_j - x_k), \,
                                \cos (x_j + x_k), \, 
                                \sin (x_j + x_k)]^T.
\end{align*}
We will use $\bs{\phi}$ and $\mbf{S} \equiv [\mbf{S}^1, \mbf{S}^2]$ to refer to the full vectors of parameters and sufficient statistics for all angles.
The natural parameter space is given by 
\begin{align*}
    \bs{\Phi} = \left\{ \bs{\phi} \, : \, \int_{[0, 2\pi)^d} \exp \left( \bs{\phi}^T \mbf{S}(\mbf{x}) \right) \, d\mbf{x} \, < \, \infty \right\}
\end{align*}
which implies that $\bs{\phi} \in \mathbb{R}^{2d^2}$ (because each angle has two marginal parameters, and each unique pair of angles has four coupling parameters, leading to $2d + 4[d(d-1)/2] = 2d^2$ parameters).

We prefer the parameterization of Equation \ref{eq:tgdiff} because it offers a simple interpretation: the sufficient statistics containing phase differences correspond to positive rotational dependence between the angles, while the sufficient statistics containing phase sums correspond to negative rotational (or reflectional) dependence. 
Positive rotational dependence occurs when phase differences are consistent across observations, that is, $X_j - X_k \approx \xi$ or $X_j \approx X_k + \xi$, for some angle $\xi$.
Then conditionally on $X_j = x_j$, $X_k$ is obtained by rotating from $x_j$ by approximately $\xi$.
Reflectional dependence instead refers to consistency in the phase sums so that $X_k \approx - X_j + \xi$, meaning that conditionally on $X_j = x_j$, $X_k$ is obtained by rotating from $-x_j$ by approximately $\xi$.
To demonstrate how each type of dependence might arise in repeated observations of neural oscillations, we show pairs of phase angles under each type of dependence in Figure \ref{fig:phasesum}; in addition, bivariate torus graph densities dominated by each type of dependence are displayed in Figure \ref{fig:bivar_posneg}.
While we have observed both kinds of dependence in neural phase angle data, rotational dependence appears to dominate in the data we analyze in this paper (see Section \ref{sec:lfpdata}).

Because the natural parameter space is $\mathbb{R}^{2d^2}$ and the $d$-dimensional torus is compact, the full $2d^2$-dimensional exponential family is regular (see \citealt[p. 2]{brown1986fundamentals}). We call the full family a {\em torus graph model} and we summarize its properties, given above, in the following theorem.

\begin{thm}[Torus graph model] \label{thm:tgnatural}
The $d$-dimensional torus graph model is a regular full exponential family. 
Equation \ref{eq:tgdiff} provides a reparameterization of the family in Equation \ref{eq:ddimprodmodel} in which the expectations of the sufficient statistics are the first circular moments and, for $d \ge 2$, the second circular moments represent rotational and reflectional dependence between pairs of variables.
In Equation \ref{eq:tgdiff}, the natural parameter has components $\bs{\phi}_j \in \mathbb{R}^2$ corresponding to the first circular moment of $X_j$ and $\bs{\phi}_{jk} \in \mathbb{R}^4$ corresponding to the second circular moments representing dependence between $X_j$ and $X_k$.
\end{thm}
We prove Theorem \ref{thm:tgnatural} in Section \ref{sec:tgmodelcomplex} by writing the angles as complex numbers and considering the complex first moments and complex-valued covariances between the variables.

Because the torus graph is an exponential family distribution with sufficient statistics corresponding to first circular moments and to pairwise interactions between variables, it is similar to the multivariate Gaussian distribution, and, as in a Gaussian graphical model, the parameters correspond to a conditional independence graph structure.
Specifically, as we state in Corollary \ref{cor:tgprop} and prove in Section \ref{sec:s3}, the pairwise coupling parameters $\bs{\phi}_{jk}$ correspond to the structure of an undirected graphical model, where an edge is missing if, and only if, the corresponding pair of variables are conditionally independent given all the other random variables.
This suggests that an undirected graphical model structure may be learned through inference on the pairwise interaction parameters, or by applying regularization in high dimensions to shrink the pairwise interaction parameters.

\begin{corollary}[Torus graph properties] \label{cor:tgprop}
The $d$-dimensional torus graph model has the following properties:
\begin{enumerate}
    \item It is the maximum entropy model subject to constraints on the expected values of the sufficient statistics.
    \item In the torus graph model, the random variables $X_j$ and $X_k$ are conditionally independent given all other variables if and only if the pairwise interaction terms involving $X_j$ and $X_k$ vanish (that is, if the entire vector $\bs{\phi}_{jk} = \mbf{0}$ in the density of Equation \ref{eq:tgdiff}).
\end{enumerate}
\end{corollary}

Another interesting property of the torus graph model, given in Theorem \ref{thm:conditional} and proven in Section \ref{sec:s6}, is that the univariate conditional distributions of one variable given the rest are von Mises, enabling Gibbs sampling to be used to generate samples from the distribution.
In addition, Theorem \ref{thm:conditional} shows torus graphs are similar to other recent work in graphical modeling in which the joint distribution is specified through univariate exponential family conditional distributions \citep{chen2014selection,yang2015graphical}.

\begin{thm}[Torus graph conditional distributions]
	\label{thm:conditional}
	Let $\mbf{X}_{-k}$ be all variables except $X_k$.
	Under a torus graph model, the conditional density of $X_k$ given $\mbf{X}_{-k}$ is von Mises; specifically,
	\begin{align*} 
		p(x_k | \mbf{x}_{-k}; \bs{\phi}) = \frac{1}{2\pi I_0(A)} \exp(A \cos(x_k - \Delta))
	\end{align*}
	where $I_m$ denotes the modified Bessel function of the first kind of $m$th order and $A$ and $\Delta$ are defined as
	\begin{align*}
	    A &= \sqrt{\left(\sum_m L_m \cos(V_m)\right)^2 + \left(\sum_m L_m \sin(V_m)\right)^2}, \\
	    \Delta &= \arctan\left(\frac{\sum_m L_m \sin(V_m)}{\sum_m L_m \cos(V_m)}\right),
	\end{align*}
	where
	\begin{align*}
	    L = \left[ \kappa_k, \bs{\phi}_{\cdot k} \right], \;
	    V = [\mu_k, \mbf{x}_{-k}, \mbf{x}_{-k}+\mbf{h}\tfrac{\pi}{2}, -\mbf{x}_{-k}, -\mbf{x}_{-k}+\tfrac{\pi}{2}],
	\end{align*}
	and $\bs{\phi}_{\cdot k}$ denoting all coupling parameters involving index $k$, $\mbf{h}_j = -1$ if $j < k$ and $\mbf{h}_j = 1$ otherwise. 
\end{thm}

%% file: tex/tgsubmodels.tex
In this section, we discuss some important subfamilies of the torus graph model that are particularly relevant to the application to neural data.
In particular, for neural phase angle data, the marginal distributions are often nearly uniform, prompting consideration of a \textit{uniform marginal model} in which the parameters $\bs{\phi}_j$ corresponding to the first moments of each variable are set to zero, resulting in a model with uniform marginal distributions.
In addition, while in our experience neural phase angle data exhibit both rotational and reflectional covariance, in many data sets, the primary form of dependence is rotational, prompting us to consider a submodel with parameters corresponding to reflectional dependence set to zero, which we will call the \textit{phase difference submodel} since all pairwise relationships are described by the sufficient statistics involving $x_j - x_k$.
A subfamily that combines these two (that is, restricts the torus graph to have marginal uniform distributions and no reflectional dependence) coincides with the models of \cite{Zemel:1993tu} and \cite{Cadieu:2010bx}, which were developed from Boltzmann machines and coupled oscillators, respectively.

The uniform marginal model, phase difference model, and a phase difference model with uniform margins all correspond to affine restrictions on the parameter space. 
This implies (see Section \ref{sec:s2}) that each is itself a regular exponential family, so that each inherits many nice properties, such as concavity of the loglikelihood function, as a function of the natural parameter.
Most previous work in multivariate circular distributions has focused on the so-called \textit{sine model} \citep[e.g.,][]{Mardia:2007kp}, which is again a subfamily, but it is not itself a full regular exponential family and does not, in general, have a concave loglikelihood function. 
As a result, estimation and inference are less straightforward than for either the torus graph model or the full regular exponential family submodels \citep{mardia2016score}. 
We summarize properties of these subfamilies in Theorem \ref{thm:tgsubmodels}, which is proven in Section \ref{sec:s2}.

\begin{thm}[Torus graph subfamilies] \label{thm:tgsubmodels}
The uniform marginal model, the phase difference model, and a model combining both parameter space restrictions, form full regular exponential families but the sine model does not.
\end{thm}

We note that the \textit{sine model} may provide parsimonious fits to data for which the marginal distributions appear unimodal. 
Even though the torus graph is a full regular exponential family, and is therefore identifiable, when the data are highly concentrated it may be hard to estimate all four coupling parameters, a phenomena we explore with simulations in Figure \ref{fig:bivkappROCMSE}. 
Neural phase angle data, however, often tend to have low concentrations while still exhibiting strong pairwise dependence (Figure \ref{fig:eda}). 
As shown in \cite{Mardia:2007kp}, when the concentration is low relative to the pairwise interaction strength, the \textit{sine model} fitted density enters a regime of multimodality. 
In Section \ref{sec:lfpdata} we demonstrate lack of fit of the \textit{sine model} to our neural data.

%% file: tex/phasecouplingintg.tex
In this section, we discusss the distinction between bivariate measures of phase coupling, such as PLV, and multivariate measures.
In Section \ref{subsec:plv}, we briefly review bivariate phase coupling measures based on the marginal distributions of pairwise phase differences.
In Sections \ref{bivtgplv} and  \ref{trivtgplv} we investigate bivariate and trivariate examples analytically.
We show that when a trivariate distribution of angles follows a torus graph, the marginal distributions of pairwise phase differences may be influenced by the marginal distributions of each variable and by indirect coupling through other nodes. This fundamental limitation of bivariate phase coupling measures can produce inaccurate phase coupling descriptions in multivariate systems.
For the special case of phase difference models with uniform margins, in Section \ref{subsec:partialplv} we propose a transformation of the torus graph parameters that produces a generalization of PLV to multivariate data (where coupling between two variables is measured conditionally on all other variables), having the nice feature that, like PLV, it falls between 0 and 1.

\subsection{Bivariate phase coupling measures} \label{subsec:plv}
The most common bivariate phase coupling measure between angles $X_j$ and $X_k$ is the Phase Locking Value (PLV) \citep{Lachaux99}, defined by
\begin{align}
	\hat{P}_{jk} &= \left|\frac{1}{N}  \sum_{n=1}^N\exp\left\{i\left(\trial{x_j}-\trial{x_k}\right)\right\} \right| 
	\label{eq:plv}
\end{align}
where $\trial{x_j}$ is the $n$th observation of $X_j$. 

We have used the notation $\hat{P}_{jk}$ to indicate it may be viewed as an estimator of its theoretical counterpart $P_{jk}$. 
In Section \ref{sec:s7}, we show that $P_{jk}$  corresponds to a measure of positive circular correlation under the assumption of uniform marginal distributions.
The value of 
$P_{jk}$ falls between 0 and 1, with 0 indicating no consistency in phase differences across trials and 1 indicating identical phase differences across trials.

One way to assess significance of $\hat{P}_{jk}$ is Rayleigh's test for uniformity of the phase differences \citep[p. 268]{kass2014analysis}; other assessments of significance typically involve permutation tests or comparison to non-task recording periods \citep{rana2013fast}.

A similar approach to characterizing bivariate phase coupling follows from considering the univariate random variable $Y_{jk} = X_j - X_k$.
If $Y_{jk}$ is distributed as von Mises with concentration parameter $\kappa$, 
then
\begin{align}
\hat{P}_{jk} = \frac{I_1(\hat{\kappa})}{I_0(\hat{\kappa})} \label{eq:plvbessel}
\end{align} 
where $\hat{\kappa}$ is the maximum likelihood estimator for $\kappa$ \citep[p. 191]{CatherineForbes:2011ub}.
More generally, any measure of the concentration of the marginal distribution of phase differences around a mean direction may be used as a measure of bivariate phase coupling \citep{Aydore:2013bb}. 
We will refer to measures based on the marginal distribution of phase differences as \textit{bivariate phase coupling measures}.

\subsection{Marginal distribution of phase differences in a bivariate torus graph} \label{bivtgplv}
Because bivariate phase coupling measures are based on the marginal distributions of phase differences, we investigate here the form of the marginal phase difference distributions in a bivariate torus graph model to determine how the torus graph parameters influence the phase differences.
For the most straightforward and analytically tractable exposition, we consider the bivariate phase difference model.
We will use the notation 
\begin{align*}
    \bs{\phi}_{jk} = [\alpha_{jk}, \beta_{jk}, \gamma_{jk}, \delta_{jk}]^T
\end{align*} 
to refer to elements of the pairwise coupling parameter vector, and use trigonometric identities to write the marginal terms as a function of $\kappa$ and $\mu$ (see Section \ref{sec:s2} for details).
Then the bivariate phase difference model density is
\begin{align*}
	p(x_1,x_2) &\propto \exp \left\{  \alpha_{12}\cos(x_1-x_2)+ \beta_{12}\sin(x_1-x_2) + \sum_{j=1}^2\kappa_j \cos(x_j-\mu_j) \right\}. 
\end{align*}
Let $W = X_1 - X_2 \, (\text{mod} \, 2\pi)$ be the phase differences wrapped around the circle so that $W \in [0,2\pi)$. 
As shown in Section \ref{sec:s4}, the unnormalized theoretical distribution of $W$ is a product of two functions:
\begin{align}
	p_W(w) &\propto f(w; \bs{\kappa},\bs{\mu}) \cdot  g(w;\alpha_{12},\beta_{12}), \label{eq:phasediffex1}
\end{align}
where 
\begin{align*}
	f(w; \bs{\kappa},\bs{\mu}) = I_0\left( \sqrt{\kappa_1^2+\kappa_2^2+2\kappa_1\kappa_2\cos(w-(\mu_1-\mu_2))} \right)
\end{align*} 
and 
\begin{align}
	g(w;\bs{\phi_{12}})=\exp \left\{ \sqrt{\left(\alpha_{12}^2+ \beta_{12}^2\right)} \cos\left(w-\arctan\left(\tfrac{\beta_{12}}{\alpha_{12}} \right) \right) \right\}. \label{eq:gfunc}
\end{align} 
The first factor, $f$, is proportional to the density of the difference of two independent von Mises random variables with concentrations $\kappa_1, \kappa_2$ and means $\mu_1, \mu_2$  \citep[p. 44]{MardiaJupp1999} and reflects the influence of the marginal distributions of $X_1$ and $X_2$ on the phase differences. 
Such convolved densities are unimodal on $[0,2\pi)$ with mode $\mu_1-\mu_2 \, (\text{mod} \, 2\pi)$ and concentration increasing with the argument of $I_0(\cdot)$.
The second factor, $g$, is proportional to a von Mises density that depends only on the phase difference and the coupling parameters.
 
The functional forms of $f$ and $g$ show that the distribution of phase differences is influenced both by the coupling parameters and by the marginal concentration parameters, which implies that bivariate phase coupling measures reflect both coupling and marginal concentration.
In Figure \ref{FigBivTGPLV}, we illustrate effects on $\hat{P}_{jk}$ of pairwise dependence and marginal concentration.
Even when the variables are independent, if the marginal distributions are not uniform, the distribution of phase differences will have nonzero concentration due to the influence of $f$. Thus, PLV is only appropriate when the marginal distributions are uniform.
In contrast, torus graph parameters can separate the influence of marginal concentration and phase coupling to provide a measure of the dependence between angles.

 \begin{figure}
 \centering
	\includegraphics[width=0.7\linewidth]{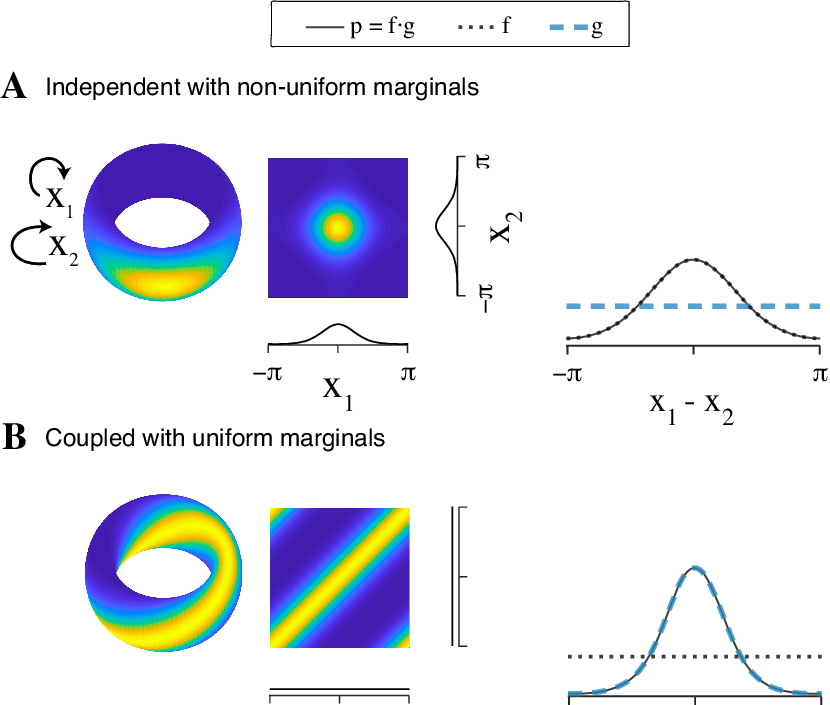}
	\caption{Examples of bivariate torus graph densities and the resulting marginal distributions of phase differences upon which bivariate phase coupling measures would be based. 
	As shown in Equation \ref{eq:phasediffex1}, 
	the density of phase differences, $p$, is affected not only by coupling (through $g$) but also by marginal concentration (through $f$). 
	As a result, bivariate phase coupling measures like PLV could give misleading results. 
	(A) Left: bivariate torus graph density with independent angles and non-uniform marginal distributions; density is shown on the torus and flattened on $[-\pi,\pi]$ with marginal densities on each axis. Right: analytical phase difference density ($p$, solid) which is a product of a direct coupling factor ($g$, dashed) and a marginal concentration factor ($f$, dotted). Here, $p$ is concentrated solely through the marginal concentration factor $f$, implying bivariate phase coupling measures would indicate coupling despite the independence of $X_1$ and $X_2$. (B) Similar to A, but with coupling between angles and uniform marginal distributions; only in this case does $p$ correctly reflect the coupling.}
	\label{FigBivTGPLV}
\end{figure}

\subsection{Marginal distribution of phase differences in a trivariate torus graph model} \label{trivtgplv}
While we have shown that torus graph models are preferable to bivariate phase coupling measures in the bivariate case, the biggest advantage of using torus graphs comes from the ability to work with multivariate data and determine unique associations between each pair of variables after conditioning on the other variables. 
For instance, in a trivariate torus graph model with direct coupling only from nodes 1 to 3 and nodes 2 to 3 (Figure \ref{FigTrivTGPLV}.A), if we were to apply bivariate phase coupling measures to all pairwise connections, we would likely infer a connection between 1 and 2 because we would be measuring the bivariate association between phase angles without taking into account node 3.

To demonstrate analytically how this happens, we consider a trivariate phase difference model with marginal concentrations equal to zero for simplicity, which has density
\begin{align}
	p(x_1,x_2,x_3)&\propto \exp \left\{ \sum_{(j,k)\in E} \begin{bmatrix} \alpha_{jk} \\ \beta_{jk} \end{bmatrix}^T 
				\begin{bmatrix} \cos (x_j - x_k) \\  \sin (x_j - x_k) \end{bmatrix} \right\}, \label{eq:unifphasediff}
\end{align}
where the edge set $E = \{(1,2),\ (1,3),\ (2,3)\}$. 
Letting $W = X_1 - X_2 \, (\text{mod} \, 2\pi)$ be the phase difference between nodes 1 and 2, we show in Section \ref{sec:s4} that the unnormalized density of $W$ is given by the product of two factors:
\begin{align} 
	p_W(w) \propto g(w;\bs{\phi_{12}}) \cdot h(w; \bs{\phi_{13}},\bs{\phi_{23}}). \label{eq:trivphasediff}
\end{align}
The first factor is the same as $g$ in Equation \ref{eq:gfunc} and reflects direct connectivity between $X_1$ and $X_2$ as it depends only on the coupling parameters for the pair, $\bs{\phi}_{12}$. 
The second factor reflects indirect connectivity through the other nodes, as it depends on the coupling parameters for the other pairs:
\begin{align*}
	h(w; \bs{\phi_{13}},\bs{\phi_{23}}) \propto I_0\left(\sqrt{s +2t\cos\left(w- u \right)}\right)
\end{align*}
where 
\begin{align*}
	s &= \alpha_{13}^2+\beta_{13}^2+\alpha_{23}^2+\beta_{23}^2, \\
	t &= \sqrt{(\alpha_{13}^2+\beta_{13}^2)(\alpha_{23}^2+\beta_{23}^2)}, \\
	u &= \arctan\left(\tfrac{\beta_{13}}{\alpha_{13}}\right)-\arctan\left(\tfrac{\beta_{23}}{\alpha_{23}}\right).
\end{align*}
Therefore, $h$ is proportional to the density of the difference of two independent von Mises random variables with concentrations $\sqrt{\alpha_{13}^2+\beta_{13}^2}$ and $\sqrt{\alpha_{23}^2+\beta_{23}^2}$ and mean directions $\arctan\left(\beta_{13}/\alpha_{13}\right)$ and $-\arctan\left(\beta_{23}/\alpha_{23}\right)$, respectively.

Equation \ref{eq:trivphasediff} implies that the density of the phase differences for one pair of variables depends on all of the coupling parameters, so a bivariate phase coupling measure such as PLV will be unable to distinguish between the effects of direct coupling and indirect coupling through other nodes.
Consequently, bivariate phase coupling measures will accurately represent the direct coupling between 1 and 2 only when there is no indirect path between 1 and 2 through the other nodes.
In the most extreme case, bivariate phase coupling measures could indicate coupling even when there are \textit{only} indirect connections between two nodes through the rest of the network.
In Figure \ref{FigTrivTGPLV}, we show examples to demonstrate how the phase difference distribution is affected by both direct and indirect connections, which may result not only in contributions to the observed phase difference concentration but also in shifts in the mean  phase difference.
This demonstrates that bivariate phase coupling measures generally reflect both direct and indirect coupling; in contrast, torus graph parameters identify direct coupling. 

\begin{figure}
\centering
   \includegraphics[width=0.7\linewidth]{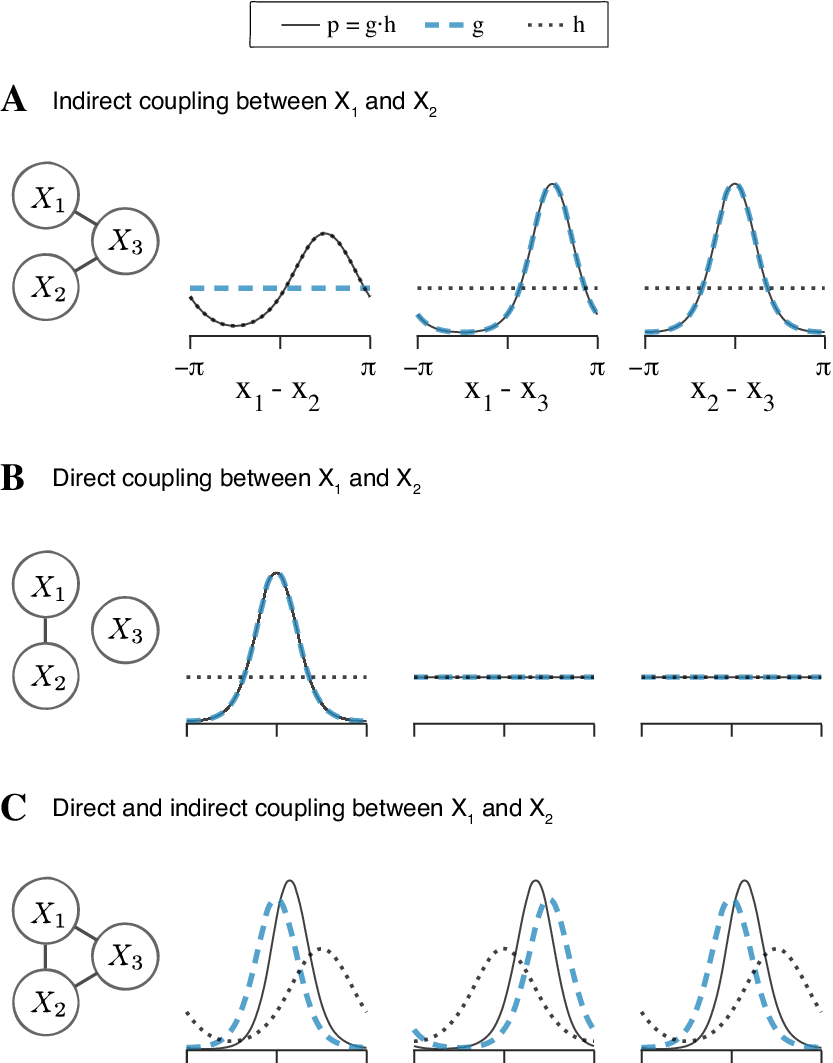}
    \caption{Examples of trivariate torus graph densities and the resulting densities of phase differences for each pair of variables. 
    As shown in Equation \ref{eq:trivphasediff}, in general, the density of phase differences, $p$, is affected not only by direct coupling (through $g$) but also by indirect connections (through $h$).  
    As a result, bivariate phase coupling measures like PLV will generally reflect both direct and indirect coupling. 
    (A) Left: ground truth graphical model, with no direct connection between $X_1$ and $X_2$ but an indirect connection through $X_3$. 
    Right: analytical phase difference densities for each pair of angles ($p$, solid) which are each a product of a direct coupling factor ($g$, dashed) and an indirect coupling factor ($h$, dotted). 
    The concentration in the phase difference $X_1 - X_2$ arises solely due to indirect connections. (B) Similar to A, but with direct connection only between $X_1$ and $X_2$; in this case, $p$ reflects the direct coupling. (C) Similar to A, but with direct connections between all nodes. Notice that indirect connections ($h$) still influence the distribution of phase differences $X_1 - X_2$ by multiplying with the direct connection term ($g$), which increases the concentration of $p$ and shifts the mean (compared to $g$).}
    \label{FigTrivTGPLV}
\end{figure}

\subsection{Interpreting phase difference model parameters} \label{subsec:partialplv}
An appealing feature of PLV is that it always falls between 0 and 1, so it is easy to interpret its magnitude and to compare PLV values for different pairs of variables.
Unfortunately, the torus graph parameters lack these qualities.
However, for the special case of the phase difference model with uniform margins, we propose a generalization of PLV based on a transformation of the torus graph model parameters that offers increased interpretability, and that, unlike PLV, measures pairwise relationships conditional on the other nodes.

As shown in Equation \ref{eq:plvbessel}, if the marginal phase difference is distributed according to a von Mises distribution, then PLV corresponds to a function of the maximum likelihood estimator of the marginal concentration parameter.
Under the phase difference model with uniform margins, we showed the marginal density of the phase difference $X_1 - X_2$ factors into terms corresponding to direct and indirect connections; the direct connectivity term $g(w; \bs{\phi_{12}})$ of Equation \ref{eq:gfunc} has the form of a von Mises density depending only on the parameters $\bs{\phi_{12}}$.
Therefore, in analogy to the definition of PLV for von Mises-distributed phase differences, we propose the following transformation of the parameters:
\begin{align*}
    \tilde{P}_{jk} = \frac{I_1\left(\sqrt{\alpha_{jk}^2 + \beta_{jk}^2}\right)}{I_0\left(\sqrt{\alpha_{jk}^2 + \beta_{jk}^2}\right)}.
\end{align*}
Like PLV, the measure always falls between 0 and 1 and therefore may be used to compare relative edge strengths in the phase difference model with uniform margins.

%% file: tex/tgestimation.tex
Because the normalization constant is intractable for the torus graph density and it cannot easily be approximated even for moderate dimension, estimation and inference are not straightforward.
In particular, maximum likelihood estimation is not feasible.
Instead, we turn to an an alternative procedure for estimation and inference called \textit{score matching}.
In Section \ref{subsec:estimation}, we establish the applicability of the score matching estimator, originally defined for densities on $\mathbb{R}^d$, to multivariate circular densities like torus graphs, then give the explicit form of the objective function and derive closed-form estimators that maximize the objective function. 
Section \ref{subsec:inference} discusses two main approaches for determining a graph structure, one based on the asymptotic distributions of score matching estimators and the second based on regularization, which is particularly relevant for high-dimensional problems.

\subsection{Estimation} \label{subsec:estimation}
Score matching is an asymptotically consistent estimation method that does not require computation of the normalization constant and is based on minimizing the expected squared difference between the model and data \textit{score functions} (gradients of the log-density functions), which leads to a tractable objective function for estimating the parameters \citep{Hyvarinen:2005wb}.
It can be seen as analogous to maximum likelihood estimation, which uses the negative log likelihood as a scoring rule; score matching instead uses the gradient and Laplacian of the log density (with respect to the data) as a scoring rule \citep{dawid2014theory}.
In addition, for real-valued exponential family distributions, the estimator comes from an unbiased linear estimating equation, so asymptotic inference is straightforward \citep{forbes2015linear,yu2018graphical}. 
However, the original score matching estimator requires the density to be supported on $\mathbb{R}^d$ and the proof of consistency relies on tail properties of such densities.
We show that score matching estimators applied to circular densities such as the torus graph model retain the same form and therefore remain consistent. 
Score matching estimators have been considered previously for the phase difference model with uniform margins \citep{Cadieu:2010bx} and the \textit{sine model} \citep{mardia2016score}, where the procedure requires modification because the \textit{sine model} is a curved exponential family distribution (Theorem \ref{thm:tgsubmodels}).

The score matching objective function is the expected squared difference between the log gradients:
\begin{align}
	\begin{split}
 	J(\bs{\phi}) &= \frac{1}{2} \int p_\mbf{X}(\mbf{x}) || \nabla_\mbf{x} \log q(\mbf{x};\bs{\phi}) - \nabla_\mbf{x} \log p_\mbf{X}(\mbf{x}) ||_2^2 \; d\mbf{x}. \label{eq:smobjdef}
	\end{split}
\end{align}
The objective function depends on the unknown data density $p_\mbf{X}(\mbf{x})$ in a nontrivial way, but we show in Theorem \ref{thm:scorematch}, using techniques similar to \cite{Hyvarinen:2005wb,Hyvarinen:2007je}, that the objective function may be simplified to depend on the data density only through an expectation, allowing it to be estimated as an average over the sample (proof in Section \ref{sec:s5}).

\begin{thm}[Score matching estimators for torus graphs]
\label{thm:scorematch}
Under some mild regularity assumptions (given in Section \ref{sec:s5}), the score matching objective function for the torus graph model takes the form
\begin{align*}
	J(\bs{\phi}) &= E_{\mbf{x}} \left\{  \frac{1}{2} \bs{\phi}^T \bs{\Gamma}(\mbf{x}) \bs{\phi} - \bs{\phi}^T \mbf{H}(\mbf{x}) \right\} 
\end{align*}
where 
\begin{align*}
    \mbf{H}(\mbf{x}) = [\mbf{S}^1(\mbf{x}), \,
 2\mbf{S}^2(\mbf{x})]^T
\end{align*}
is a vector with dimension $2d^2$ that is a simple function of the sufficient statistics and
\begin{align*}
    \bs{\Gamma}(\mbf{x}) = \mbf{D}(\mbf{x}) \mbf{D}(\mbf{x})^T
\end{align*}
where
\begin{align*}
    \mbf{D}(\mbf{x}) = \nabla_\mbf{x} \mbf{S}(\mbf{x})
\end{align*} 
is the $2d^2 \times d$ Jacobian of the sufficient statistic vector.
Specific expressions for the Jacobian elements are given in Section \ref{sec:s5}.
\end{thm}

Theorem \ref{thm:scorematch} shows that the score matching objective may be estimated empirically by averaging over $N$ observed samples.
The empirical objective function is
\begin{align}
 	 \tilde J(\bs{\phi}) =\frac{1}{2} \bs{\phi}^T \hat{\bs{\Gamma}} \bs{\phi} - \bs{\phi}^T \hat{\mbf{H}}  \label{eq:smempirical}
\end{align}
where 
\begin{align*}
    \hat{\bs{\Gamma}}  = \frac{1}{N} \sum_{n=1}^N \bs{\Gamma}\left(\trial{\mbf{x}}\right), \; \hat{\mbf{H}} = \frac{1}{N} \sum_{n=1}^N  \mbf{H}\left(\trial{\mbf{x}}\right)
\end{align*}
with $\trial{\mbf{x}}$ denoting sample $n$.
Taking the derivative of Equation \ref{eq:smempirical} with respect to the parameter vector yields an unbiased estimating equation \citep{dawid2014theory}, which has a unique solution when $\hat{\bs{\Gamma}}$ is invertible:
\begin{align*}
	\hat{\bs{\Gamma}} \bs{\phi}  - \hat{\mbf{H}} = 0 \; \longrightarrow \; \hat{\bs{\phi}} = \hat{\bs{\Gamma}}^{-1} \hat{\mbf{H}}.
\end{align*}
The number of parameters for a $d$-dimensional torus graph is $2 d^2$ so sample sizes may not be sufficient for $ \hat{\bs{\Gamma}}$ to be invertible.
In particular, $ \hat{\bs{\Gamma}}$ is a sum with $Nd$ terms, so $N$ must be greater than $2d$ for $ \hat{\bs{\Gamma}}$ to be invertible.
In practice, the variance of estimated parameters will be high if $N$ is not much larger than $2d$, leading to less accurate point estimates and inference.
We investigate the effect of sample size on the resulting inferences, using simulated data, in Section \ref{sec:simstudy}.

For higher-dimensional problems, Equation \ref{eq:smempirical} is a convex objective function that may be minimized numerically with regularization.
In torus graphs, a group $\ell_1$ penalty may be placed on the groups of pairwise coupling parameters $\bs{\phi}_{jk}$ to enforce sparsity in the estimated edges, yielding the following objective function:
\begin{align*}
	\tilde{J}_\lambda(\bs{\phi}) =  \tilde{J}(\bs{\phi}) + \lambda \sum_{j < k} || \bs{\phi}_{jk} ||_2.
\end{align*}
Here, $\lambda$ is a tuning parameter that may be selected by criteria such as cross-validation or extended BIC \citep{Lin:2016wq}. 
Other structured penalties may be used to encourage the model toward specific submodels (such as the phase difference model or the uniform marginal model).
For instance, separate group $\ell_1$ penalties could be applied to the pairs of coupling parameters corresponding to positive and negative dependence, or an $\ell_2$ penalty on the marginal parameters $\bs{\phi}_j$ could encourage low concentration.
This type of penalization may improve behavior of the objective function, and could be especially useful when certain subfamilies appear reasonable based on exploratory data analysis.
The computational burden of calculating $\hat{\bs{\Gamma}}^{-1}$ may be reduced using the conditional independence structure of the graph, as we may estimate each four-dimensional group of parameters $\bs{\phi}_{jk}$ using score matching on the conditional distribution $p\left(x_j, x_k | \mbf{x}_{-jk} \right)$ ,
which involves only $8(d-1)$-dimensional sufficient statistics and thus uses only a subset of the rows of $\mbf{D}(\mbf{x})$, lessening the computational burden of matrix inversion \citep{yu2016statistical}.

\subsection{Inference} \label{subsec:inference}
In our setting, the goal of inference is to determine a graph structure by determining which pairs of variables $\{j,k\}$ have nonzero $\bs{\phi}_{jk}$, indicating an edge between nodes $j$ and $k$.
As shown in previous work, score matching estimators are asymptotically normal \citep{dawid2014theory,forbes2015linear,yu2018graphical}, that is, 
	\begin{align}
		\sqrt{N} \left( \bs{\hat{\phi}} - \bs{\phi} \right) \xrightarrow{d} \mathcal{N}\left(\mbf{0}, \mbf{\Sigma} \right) \label{eq:asympdist}
	\end{align}
where the asymptotic variance is given by
\begin{align*}
    \mbf{\Sigma} = \bs{\Gamma}_0^{-1}
    \mbf{V}_0
    \bs{\Gamma}_0^{-1}
\end{align*}
with
\begin{align*}
    \bs{\Gamma}_0 = E[\bs{\Gamma}(\mbf{x})], \; \mbf{V}_0 = E[(\bs{\Gamma}(\mbf{x}) \bs{\phi} - \mbf{H}(\mbf{x}))(\bs{\Gamma}(\mbf{x}) \bs{\phi} - \mbf{H}(\mbf{x}))^T].
\end{align*} 
Sample averages may be substituted for the expectations to obtain an estimate of the asymptotic variance, and because the true value of $\bs{\phi}$ is unknown, we may substitute either our estimate $\hat{\bs{\phi}}$ or a null hypothetical value. 

By considering the marginal Gaussian distribution of each element of $\bs{\phi}$, confidence intervals may be constructed in a standard way.
However, in the torus graph model, there are four parameters per edge, so individual parameters are not of primary interest.
In addition, we may be interested in testing hypotheses about groups of edges (for example, the null hypothesis might be that there are no edges between regions A and B). 
Fortunately, inference on groups of edges is also straightforward, as specified in the following lemma.
\begin{lemma}[Asymptotic distribution for groups of torus graph parameters.] \label{lemma:chisq}
        A vector of parameters indexed by an index set $E$ of size $|E|$, denoted $\bs{\phi_E}$, satisfies 
	\begin{align*}
		\sqrt{N} \left( \bs{\hat{\phi}_E} - \bs{\phi_E} \right) \xrightarrow{d} \mathcal{N} \left(\mbf{0}, \mbf{\Sigma_E} \right)
	\end{align*}
	where $\mbf{\Sigma_E}$ is the corresponding submatrix of the asymptotic variance $\mbf{\Sigma}$ of Equation \ref{eq:asympdist}. 
	We also have
	\begin{align*} 
		N \left(  \bs{\hat{\phi}_E} - \bs{\phi_E} \right)^T \mbf{\Sigma_E}^{-1} \left(  \bs{\hat{\phi}_E} - \bs{\phi_E} \right) \xrightarrow{d} \chi^2(|E|).
	\end{align*} 
\end{lemma}
Lemma \ref{lemma:chisq} enables computation of $p$-values for single edges 
(if $E$ indexes the four parameters for a single edge) or for groups of edges. 
In particular, if $E$ indexes the four parameters corresponding to a single edge, then under the null hypothesis that $\bs{\phi_E} = 0$, 
	\begin{align*}
		X^2_E \equiv N \bs{\hat{\phi}_E}^T \mbf{\Sigma_E}^{-1} \bs{\hat{\phi}_E} \xrightarrow{d} \chi^2(4),
	\end{align*} 
so for an observed value of the test statistic $\hat{X}^2_E$, the probability statement $P\left(X^2_E  \ge \hat{X}^2_E \right)$, 
which gives a $p$-value for the edge, may be evaluated using the $\chi^2$ distribution with 4 degrees of freedom.
Similarly, a $\chi^2$ distribution with two degrees of freedom may be used to test for only rotational or only reflectional covariance, or to test for nonzero marginal parameters.

Inference after regularization is less straightforward.
Recent work has addressed inference for score matching estimators when using an $\ell_1$ penalty on each parameter \citep{yu2016statistical}, which could potentially be extended to torus graphs with a group $\ell_1$ penalty.
Other approaches for inference in high dimensions include the bootstrap or stability selection \citep{meinshausen:2010jg}.
One parametric bootstrap approach is as follows.
Assume we are interested in testing the null hypothesis that some particular subset of edges is missing from the graph.
We may first fit a null torus graph model in which the coupling parameters corresponding to the edge set of interest are set to zero, selecting the regularization parameter by cross-validation of the score matching objective function.
Next, $B$ times, we would draw samples of the same size as the data from the null torus graph model, re-select the regularization parameter by cross-validation, and fit the unrestricted torus graph model to the samples using this regularization parameter.
By computing the distribution of a suitable statistic (such as the number of nonzero edges or the maximal edgewise parameter vector norm) from these fitted null models, we obtain an empirical estimate of the null distribution of the statistic, which can then be used to judge the size of the same statistic computed on the original data.

%% file: tex/simstudy.tex
As our analytical results of Section \ref{sec:tgphasecoupling} show, torus graphs can separate the effects of pairwise coupling and marginal concentration and have pairwise coupling parameters that represent direct connections between nodes.
In contrast, bivariate phase coupling measures like PLV are sensitive to the marginal distribution of the variables and can reflect not only direct connections but also indirect paths through the other nodes.
We conducted simulations to demonstrate these results.
In addition, we explored the performance of torus graphs in recovering graph structures in simulated data similar to real data to determine how well we expect torus graphs to perform in the real data.
Section \ref{subsec:sim} gives the simulation details and Section \ref{subsec:simresults} provides the results.

\subsection{Simulation methods} \label{subsec:sim}
When comparing PLV to torus graphs, we chose to generate data using the notion of positive rotational dependence discussed in Section \ref{sec:tgmodel}.
This was done to demonstrate that torus graphs recover interactions of this type even when data were not directly generated from a torus graph model.
To generate bivariate data with rotational dependence and nearly uniform marginal distributions, we first drew $N$ trials of phase angles $x_1$ from a von Mises distribution with low circular concentration $\kappa_1$. 
Then, for each trial, we let $x_2 = x_1 + \xi + \epsilon$ where $\xi$ is a fixed phase offset and $\epsilon$ is noise drawn from a concentrated mean-zero von Mises distribution with concentration $\kappa_\epsilon$ (where, on a small number of trials, we used less concentrated noise to emulate the noisiness present in real data).
Extending to more than two nodes follows a similar process, where data for an additional node is generated based on data from a neighbor in the graph.

We generated synthetic data with two ground truth phase coupling structures that are intended to reflect realistic scenarios (and with parameters chosen to produce samples that emulate real neural phase angle data; Figure \ref{fig:toyvalidation} compares the simulated data to real data, showing similar first- and second-order behavior and similar observed pairwise PLV values).
First, we constructed five-dimensional data meant to emulate the effects of spatial dependence, such as dependence between electrodes on a linear probe situated within a single functional region, which, under a nearest-neighbor Markov assumption, would induce sparse conditional independence graph structures (because each node would be directly dependent only on its nearest neighbors on the probe).
We coupled nodes in a linear chain and chose $\xi = \pi/100$ and $\kappa_\epsilon = 40$ (with 15 of 840 trials contaminated with extra noise with concentration 0.1).
Second, we constructed three-dimensional data meant to emulate the effects of indirect connections, which may occur when electrodes are in different regions, but not all of the regions are communicating directly.
In particular, $x_2$ was concentrated at $\kappa_2 = 0.01$, $x_1$ and $x_3$ had phase offsets of $\xi = \pi/6$ and $\xi = \pi/100$, respectively, from $x_2$, and the coupling noise had concentration $\kappa_\epsilon = 2$ (with 75 of 840 trials contaminated with extra noise with concentration 0.1). 
For each scenario, we simulated data of sample size 840 (to match the sample size of the real data).
Then, for each data set, we fitted a torus graph and selected edges based on Lemma \ref{lemma:chisq}; we also used the Rayleigh test of uniformity to construct a graph based on PLV \citep[p. 268]{kass2014analysis}.
For both tests, we used an alpha level of $\alpha = 0.001$ with Bonferroni correction for multiple tests.

To gain intuition on how well torus graphs could be expected to perform in the real LFP data we analyze in Section \ref{sec:lfpdata}, we investigated how well torus graphs recover the edges for varying dimensions, sample sizes, and underlying levels of sparsity in the edges .
For this simulation, we generated data from a torus graph model of varying dimension with zero marginal concentration and with either 25\% or 50\% of edges present in the generating distribution.
By varying the threshold on the edgewise $X^2$ statistics (Lemma \ref{lemma:chisq}), we computed an ROC curve for each simulated data set.
The ROC curves were averaged across 30 data sets, then the area under the curve (AUC) was calculated as a measure of performance.

\subsection{Simulation results} \label{subsec:simresults}
For the first set of simulations, Figure \ref{FigToySims} shows that in both the three-dimensional and five-dimensional cases, the torus graph recovered the correct structure while the PLV graph recovered a fully connected graph.
Although the performance of PLV may be better for other graph structures, our analytical results in Section \ref{sec:tgphasecoupling} suggest that graph structures with indirect paths between nodes are likely to induce excess edges in the PLV graph.
To follow up on this result, we further explored the False Positive Rate (FPR) and False Negative Rate (FNR) for PLV and torus graphs by repeating the simulations. 
We found that PLV graphs have low FNR (near 0), but also have a very high FPR (near 1), so PLV likely won’t miss a true edge but will also add many additional edges (Figure \ref{fig:fnrfpr}). 
This result agrees with the notion that hypothesis testing based on PLV, even when corrected for multiple comparisons, cannot be reliably used to control FPR for multivariate graphs because PLV measures both direct as well as indirect connectivity and thus tends to overestimate connectivity. 
In contrast, torus graphs are more conservative in assigning edges and control the FPR at the nominal level (though they tend to have higher FNR, especially for low sample sizes).

\begin{figure}
    \centering
    \includegraphics[width=0.7\linewidth]{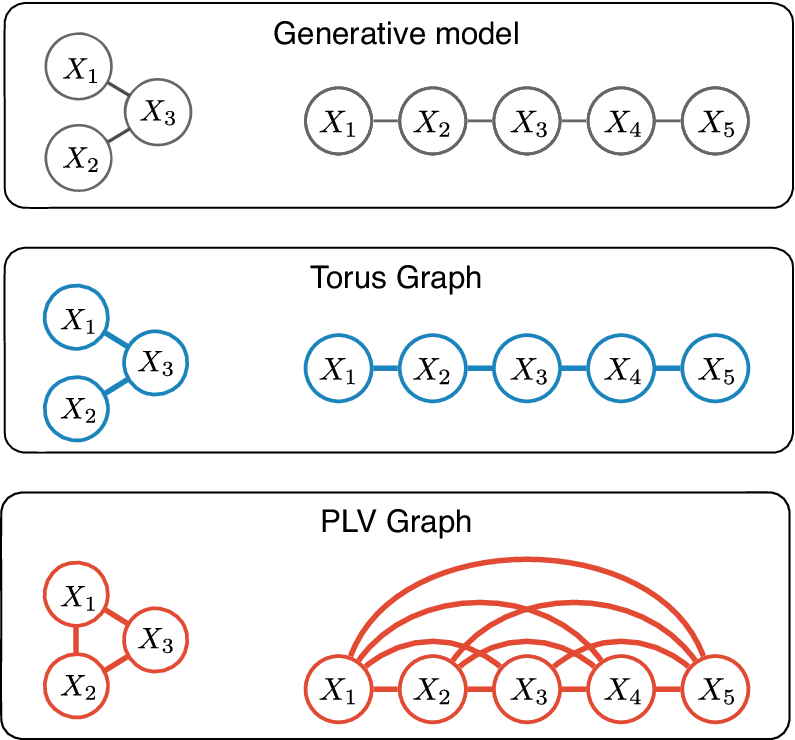}
    \caption{The torus graph recovers the ground truth graph structures (top panel) from realistic simulated data sets while a bivariate phase coupling measure, phase locking value (PLV), does not (edges shown for corrected $p < 0.001$). 
    Left: a 3-dimensional simulated example of cross-area phase coupling where regions $X_1$ and $X_2$ are not directly coupled, but are both coupled to region $X_3$. 
    Right: a 5-dimensional simulated example of a graph structure that could be observed for channels on a linear probe with nearest-neighbor spatial dependence.
    In both cases, PLV infers a fully-connected graph due to indirect connections.}
    \label{FigToySims}
\end{figure}

Figure \ref{fig:AUCbydensity} displays the results of the second set of simulations, which investigated the ability of torus graphs to recover the true structure as a function of true edge density, sample size, and data dimension. 
Importantly, for simulated data of dimension 24 and sample size of 840 (matching the real LFP data), the torus graph model is able to achieve 0.9 AUC as long as the graph is sufficiently sparse (around 25\% of all possible edges present).
In the real data results of Figure \ref{fig:AcrossROI}.B, we in fact observe approximately 25\% of edges present, suggesting that this graph density may be reasonable for the real data.
A more detailed investigation of the ROC curves and precision curves by dimension with fixed sample size 840 is given in Figure \ref{fig:ROCprecbydim}, which again demonstrates that for a sufficiently sparse underlying graph structure, the torus graph method is expected to perform well for the sample size and dimension in the real LFP data.
However, prior beliefs about the sparsity of the underlying graph will play a role in judging the likely accuracy of results.

\begin{figure}
    \centering
    \includegraphics[width=0.7\linewidth]{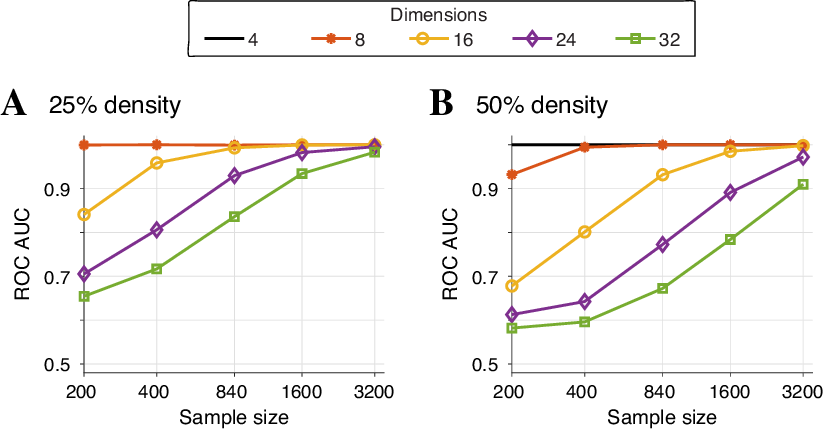}
    \caption{
    In simulated data with two different underlying edge densities, the average ROC curve area under the curve (AUC) was computed across 30 simulated data sets as a function of sample size (shown along the horizontal axis). 
    The dimension of the data is indicated by line color and different markers. 
    Panel A demonstrates that if the true underlying graph has only 25\% of all possible edges present, then even for 24 dimensional data (diamond markers), a sample size of 840 (the size of our real LFP data set) is sufficient to reach AUC above 0.9. While performance degrades when the underlying graph is more dense, panel B shows that performance is still reasonable (AUC near 0.8) for 24 dimensional data with 840 samples.}
    \label{fig:AUCbydensity}
\end{figure}

%% file: tex/lfpdata.tex
We demonstrate torus graphs in a set of local field potentials (LFPs) collected from 24 electrodes in the prefrontal cortex (PFC) and hippocampus (HPC) of a macaque monkey during a paired-associate learning task. Previous analysis of these data in \cite{Brincat:2016cl} found that beta-band (16 Hz) phase coupling between PFC and HPC peaked during the cue presentation and also increased with learning after the subject received feedback on each trial.
Here, we sought a more fine-grained description of the phase coupling between PFC and HPC during the cue presentation period, and focused on describing relationships between PFC and three distinct subregions of HPC: subiculum (Sub), dentate gyrus (DG), and CA3.

First, we applied torus graphs to two different low-dimensional subnetworks: (i) a subnetwork consisting of five electrodes arranged linearly along a probe within CA3 and (ii) a collection of all trivariate subnetworks consisting of an electrode in each of the regions Sub, DG, and PFC.
The five-dimensional subnetwork was chosen as a proof of concept, because  electrodes in the same region and with a linear spatial arrangement ought to exhibit a nearest-neighbor conditional independence structure. 
We chose to examine connectivity between Sub, DG, and PFC because the patterns of connectivity between these three regions could be informative about whether hippocampal activity is leading prefrontal activity and because torus graphs should be able to disentangle the effect of direct and indirect connections to give a more informative connectivity structure than could bivariate phase coupling measures.
Second, we applied torus graphs to the full 24-dimensional data set by first testing for the presence of any cross-region edges between PFC, Sub, DG, and CA3, and then following up with post-hoc tests of individual cross-region and within-region edges to construct a full 24-dimensional graph.
Finally, we used a subset of the PFC electrodes to examine the goodness-of-fit of the torus graph model to the data and to investigate whether any torus graph subfamilies appeared to be appropriate for this data set.

We describe the data and preprocessing in Section \ref{subsec:data} and give an outline of our data analysis methods in Section \ref{subsec:dataanalysis}.
Section \ref{subsec:dataresults} presents the results and we discuss implications of the results in Section \ref{subsec:datadiscussion}.

\subsection{Experiment and data details} \label{subsec:data}
The experimental design and data collection procedures are described thoroughly in \cite{brincat2015frequency, Brincat:2016cl}.
We use data from a single animal in a single session comprising 840 trials in which a correct response was given. 
(The sample size here is 840; a very small number of animals, usually 1 or 2, is standard practice in nonhuman primate neurophysiology because, even though there is large subject-to-subject variability in the fine details of brain structure and function, the overall structure and function of major brain regions is conserved, as are, typically, the primary scientific conclusions, though it is common to replicate in a second animal results found in a single animal; we also note that while part of the purpose of the original experiment involved learning, we are here ignoring any transient learning effects, which take place rapidly.)
Briefly, four images of objects were randomly paired; the subject learned the associations between pairs through repeated exposure to the pairs followed by a reward for correctly identifying a matching pair. 
In each trial, the pairs of images were presented sequentially with a 750 ms delay period between the images, during which a fixation mark was shown. 
All procedures followed the guidelines of the MIT Animal Care and Use Committee and the US National Institutes of Health. 
(The experimental procedures were painless to the animals, as all forms of sensation originate outside the brain.) 
The data used in this paper contain 8 single-channel electrodes in PFC and a linear probe with 16 channels in HPC, with HPC channels categorized based on neural spiking characteristics into three subregions: dentate gyrus (DG), CA3, and subiculum (Sub). 
Recording regions and data processing steps are depicted in Figure \ref{fig:repeatedtrials}.
We focus on a time point at 300ms after initial cue presentation, as PLV pooled across all sessions identified phase coupling peaking near 16Hz at this time point \citep[Supplementary Figure 5]{Brincat:2016cl}; we verified that the session we used showed the same overall phase coupling relationship. 
After downsampling the data to 200Hz and subtracting the average (evoked) response, we used complex Morlet wavelets with 6 cycles to extract the instantaneous phase of each channel at 16Hz in each trial.

\begin{figure}
    \centering
    \includegraphics[width=0.7\linewidth]{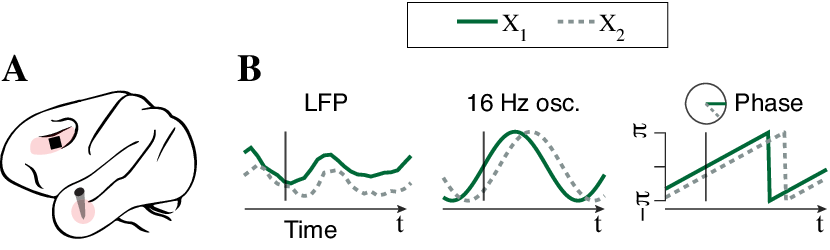}
    \caption{(A) Depiction of recording sites in ventrolateral prefrontal cortex (PFC) and hippocampus. (B) Preprocessing to obtain phase angles: local field potential (LFP) signals are filtered using Morlet wavelets to extract phase angles from 16 Hz oscillations at a time point of interest (two signals are shown for a single trial; repeated observations of phase angles are collected across repeated trials).}
    \label{fig:repeatedtrials}
\end{figure}

\subsection{Data analysis methods} \label{subsec:dataanalysis}
To examine whether torus graphs could recover the spatial features we would expect along the linear probe, we first applied torus graphs to the 5-dimensional network containing channels on a linear probe that are all within CA3, likely to exhibit strong spatial dependence between neighboring channels, and used a hypothesis test for each edge with $\alpha = 0.0001$.
We chose a stringent threshold because, based on the first simulation study of Section \ref{sec:simstudy}, we expected PLV to add extraneous edges, yet we wanted to demonstrate that torus graphs and PLV give very different results even when a small threshold is used.
Then, to examine whether torus graphs appeared to disentangle the effects of direct and indirect edges, we focused first on a trivariate network containing Sub, DG, and PFC where there is a simple interpretation of direct edges because CA3 and Sub send output signals from hippocampus while DG receives input signals to the hippocampus.
Therefore, prominent connections between CA3 and PFC and/or Sub and PFC would suggest hippocampal activity may be leading PFC activity during this period of the task, while dominance of connections between DG and PFC would suggest the opposite.
Because there are multiple channels in each of the three regions, we aggregated results across all possible triplets of channels from the three regions by inferring edges by majority vote across all possible trivariate graphs (using an alpha level of $p < 0.0001$ for each edgewise test).

We also investigated the graphical structure between the four regions (PFC, DG, CA3, and Sub) using a 24-dimensional torus graph on all electrodes; based on model selection techniques described in the next paragraph, we used a phase difference model.
First, we assessed between-region connectivity between each pair of regions using the results of Lemma \ref{lemma:chisq} to test the set of null hypotheses that there were no edges between each pair of regions.
For example, there are 40 total possible edges between CA3 and PFC with two parameters for each edge, so the hypothesis test for the existence of any edges between CA3 and PFC was based on a $\chi^2(80)$ distribution.
For each pair of regions, we obtained a $p$-value for the entire group of edges between the two regions under the null hypothesis that there are no edges between the two regions.
Because we were interested only in further investigating cross-region interactions with strong evidence, for this set of hypothesis tests, we chose a stringent alpha level of $\alpha = 0.001$ with a Bonferroni correction across all 6 between-region tests to control for multiple tests.
To better understand the individual connections driving this cross-region connectivity, and to investigate within-region connectivity patterns, we then applied post-hoc tests on the individual edges.
That is, for each possible edge between pairs of regions that were identified as having some connection by the first step, we obtained a $p$-value using a $\chi^2(2)$ distribution. 
Similarly, we calculated a $p$-value for each possible edge connecting electrodes within the same region.
In this case, we assigned edges using a less stringent alpha level of $\alpha = 0.05$ without correcting for multiple tests, as we expected the evidence for any specific edge would be weaker than the evidence for cross-region connections and because multiple comparisons were already taken into account in the first set of between-region tests. 
 
To assess the appropriateness of torus graphs for this data set and to assess whether any submodels were appropriate, we explored the first- and second-order behavior of the LFP phase angles.
To determine whether the uniform marginal model should be fitted, we tested for uniform marginal distributions using a Rayleigh test on each electrode, marginally, then obtained an overall decision regarding the null hypothesis that all distributions were uniform using Fisher's method to combine $p$-values \citep[p. 301]{kass2014analysis}.
Equation \ref{eq:trivphasediff} showed that the marginal distribution of phase differences in a torus graph model depends on the coupling parameters for all possible direct and indirect paths between two nodes, so that any observed concentration of phase differences, marginally, could indicate the presence of \textit{some} nonzero coupling parameters (possibly corresponding to indirect paths).
Therefore, we applied the Rayleigh test to the observed phase differences for each pair of variables, then combined $p$-values using Fisher's method to test the null hypothesis that all marginal phase difference distributions were uniform; we also performed a similar procedure on the observed pairwise phase sums.
For these exploratory tests, we used a $p$-value threshold of $p < 0.05$ so that we would be sensitive to departures from uniformity in either the marginal distributions or the distributions of phase differences/sums.
Finally, after fitting the chosen model, we used Kolmogorov-Smirnov (KS) goodness-of-fit tests to determine whether there was any evidence that LFP angles were drawn from a different distribution than the fitted theoretical model.
In particular, we tested for differences in the marginal distributions of the angles, then tested for differences in the marginal distributions of pairwise phase differences and pairwise phase sums, and used Fisher's method to combine $p$-values within each of the three groups of tests.
We used an alpha level of $0.05$ for each test.

\subsection{Data results} \label{subsec:dataresults}
In the low-dimensional subnetworks, we found that the torus graph recovered a structure consistent with nearest-neighbor coupling along the linear probe (Figure \ref{FigLFPlowdim}, top panel), while PLV suggested a fully-connected graph; Figure \ref{fig:three_five_var}B shows $p$-values for the five-dimensional graph for both torus graphs and PLV.
In addition, the top panel of Figure \ref{FigLFPlowdim} shows that torus graphs appear to capture an interesting trivariate network with no coupling between PFC and DG, but with each of PFC and DG coupled with Sub.
Figure \ref{fig:three_five_var}C shows $p$-values for all of the individual trivariate graphs, indicating that in the majority of individual trivariate graphs, there was no evidence for an edge between PFC and DG, while again, PLV suggested a fully connected graph.

\begin{figure}
\centering
	\includegraphics[width=0.7\linewidth]{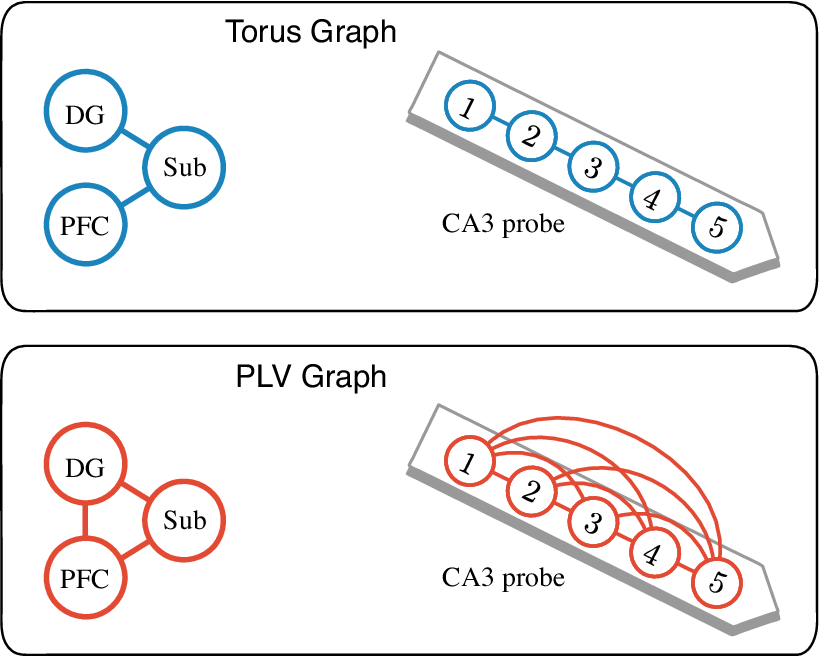}
	\caption{Torus graphs and PLV graphs from low-dimensional networks of interest in LFP data. Left: cross-region connectivity between dentate gyrus (DG), subiculum (Sub), and PFC, where the torus graph (top panel) indicated that DG and PFC are each coupled to Sub; in contrast, PLV (bottom panel) inferred a fully-connected graph. Right: within-region connectivity in CA3, where the torus graph indicated a spatial dependence structure which reflects the placement of channels along a linear probe, while PLV inferred a fully-connected graph.}
	\label{FigLFPlowdim}
\end{figure}

For the 24-dimensional torus graph applied to all electrodes, we found, first, using the overall tests for the presence of any edges between each pair of regions, that there was apparent connectivity between all hippocampal subregions (CA3, DG, and Sub) and connectivity from CA3 to PFC and Sub to PFC (Figure \ref{fig:AcrossROI}.A).
In the follow-up post-hoc tests of individual edges, we observed dense connectivity within regions and somewhat sparser connectivity between regions (as judged by the number of edges out of the possible number that could be present; graph and adjacency matrix shown in Figure \ref{fig:AcrossROI} panels B and C).
Interestingly, none of the individual CA3 to Sub connections were significant (the smallest $p$-values were around 0.1), suggesting that the aggregate effect of several weak edges led to the edge between CA3 and Sub in Figure \ref{fig:AcrossROI}.A. 
In the adjacency matrix corresponding to the 24-dimensional graph, which is ordered to respect the position of channels on the hippocampal linear probe, entries are colored by $p$-value (with white indicating non-significant entries at level $\alpha = 0.05$); notably, torus graphs recovered the linear probe structure across the entire hippocampus without prior knowledge of this structure being used in the model or estimation procedure.

\begin{figure}
\centering
	\includegraphics[width=0.7\linewidth]{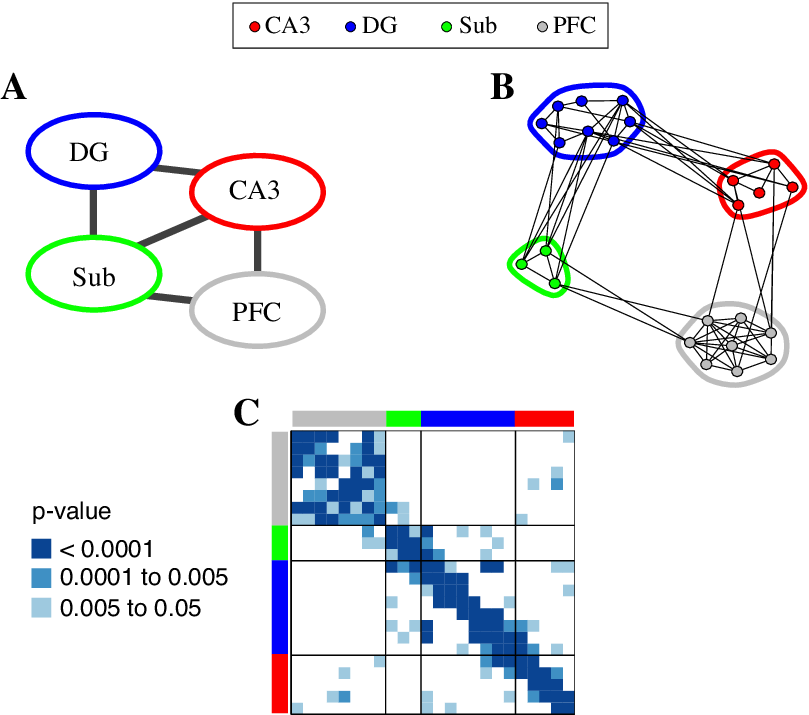}
	\caption[foo bar]{
	Torus graph analysis of coupling in 24-dimensional LFP data with four distinct regions: CA3 (red), dentate gyrus (DG; blue), subiculum (Sub; green), and prefrontal cortex (PFC; grey). 
	(A) Cross-region tests indicated evidence for edges between DG and CA3, DG and Sub, CA3 and Sub, PFC and CA3, and PFC and Sub (determined by testing the null hypothesis that there were no edges between a given pair of regions, corrected $p < 0.001$). 
	(B) For the significant cross-region connections, a post-hoc edgewise significance test examined the specific connections between regions, with the 24-dimensional graph containing edges for $p < 0.05$. 
	Compared to the cross-region graph in A, notice that no edges from Sub to CA3 were individually significant at $p < 0.05$. 
	(C) The 24-dimensional adjacency matrix (with hippocampal electrodes ordered by position on linear probe) with non-significant cross-region connections in white and other entries colored by edgewise $p$-value (with $p > 0.05$ in white). 
	Despite having no built-in knowledge of spatial information, the torus graph recovered the linear probe structure within hippocampus.}
	\label{fig:AcrossROI}
\end{figure}

In our investigation of the appropriateness of submodels and goodness-of-fit, we show results for three electrodes from PFC (results on the full data set were similar).
Based on a visual analysis of the marginal and pairwise behavior of the data, the phase difference model appeared to be a reasonable candidate for these data.
We found evidence that neither the marginal distributions nor the phase differences were uniform (Rayleigh/Fisher's method, $p < 0.001$); however, there was no evidence for concentration in the phase sums (Rayleigh/Fisher's method, $p > 0.05$).
As a result, we selected the phase difference model to investigate goodness-of-fit.
In Figure \ref{fig:eda}.A, we show data from three PFC electrodes along with their theoretical (fitted) phase difference model; along the diagonal, histograms for the real data (blue bins) are shown with theoretical model densities (solid red traces), indicating similar marginal distributions. 
Below the diagonal are two-dimensional histograms from the real data and above the diagonal are two-dimensional theoretical model densities, demonstrating that the torus graph appears capable of accurately representing the first- and second-order behavior present in the real data. 
Figure \ref{fig:eda}.B shows histograms of sufficient statistics from the real data (phase angles along the diagonal, phase sums for each pair above the diagonal, and phase differences for each pair below the diagonal), along with kernel density estimates (solid red traces) for the statistics from the theoretical model.
There is visual similarity between the distributions of sufficient statistics, and, in both cases, we observed the concentration of the pairwise phase differences indicating a prevalence of rotational dependence; this pattern held in the observed sufficient statistics for all 24 LFP channels shown in Figure \ref{fig:suppsuffstateda}.
The KS tests comparing the sufficient statistics of the fitted torus graph model to the data failed to reject the null hypothesis (KS/Fisher's method, $p > 0.05$, for each group of statistics), indicating no evidence that the data were drawn from a different distribution than the fitted model.

\begin{figure}
    \centering
    \includegraphics[width=0.8\linewidth]{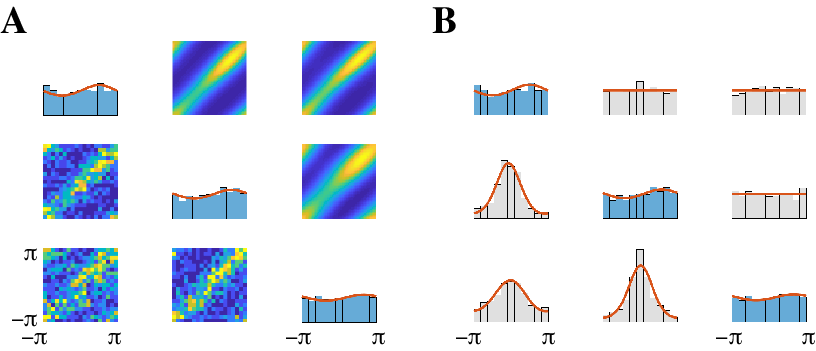}
    \caption{
    Comparison between phase angles from three LFPs located in PFC and the theoretical torus graph distribution demonstrates that torus graphs capture the salient first- and second-order behavior present in the LFP phase angles. 
    In contrast, the sine model fails to fit the data accurately Figure \ref{fig:edawithsinemodel}. 
    (A) Along the diagonal are the marginal distributions of the phase angles. 
    The real data are represented by blue histograms and the theoretical marginal densities from the torus graph model are overlaid as solid red traces. 
    Two-dimensional distributions (off-diagonal) show bivariate relationships, with theoretical densities above the diagonal and real data represented using two-dimensional histograms below the diagonal. 
    (B) Plots along the diagonal same as panel A. 
    Below the diagonal are distributions of pairwise phase differences and above the diagonal are distributions of pairwise phase sums, represented by histograms for the real data and by solid red density plots for the theoretical torus graph model. 
    Both the real data and theoretical distributions exhibit concentration of phase differences but not phase sums, suggesting prevalence of rotational covariance, and the sufficient statistics are very similar for the theoretical and real data.}
    \label{fig:eda}
\end{figure}

\subsection{Summary and discussion of results} \label{subsec:datadiscussion}
In the 24-dimensional analysis, PFC to hippocampal connections appeared to be driven by a relatively small number of significant connections from Sub to PFC and CA3 to PFC.
However, the results did not show evidence for connections between PFC and DG, which coincides with the analysis of the trivariate subnetwork.
In contrast, the PLV edgewise adjacency matrix (shown in Figure \ref{fig:plvadj}) was very densely connected and shows little resemblance to the structure recovered by torus graphs; for any reasonable $p$-value threshold, the PLV graph would be nearly fully connected, suggesting that PLV was unable to distinguish between direct and indirect connections.

Because the patterns of connectivity between hippocampal subregions and PFC recovered by these analyses suggested that hippocampal activity was leading prefrontal activity during this time period in the task, as a follow-up analysis, we considered whether lead-lag relationships could be detected in the distribution of phase differences between PFC and hippocampus (specifically, CA3 and Sub).
Since most of the edgewise dependence in this data set appears to correspond to positive rotational dependence (as judged by the relative magnitude of the torus graph coupling parameters and the concentration of phase differences but not phase sums in the observed sufficient statistics shown in Figure \ref{fig:suppsuffstateda}), the distributions of phase differences between PFC and Sub and between PFC and CA3 can summarize the overall PFC-hippocampus coupling.
Figure \ref{fig:phasediff}.A displays a circular histogram of the PFC to hippocampus phase differences with the mean phase difference and a 95\% confidence interval, pooling across all significant ($p < 0.05$) pairwise phase differences to compute an overall circular mean phase offset.
These phase differences were centered at $-8.5^\circ$ (95\% CI: $[-10.4^\circ,-6.6^\circ]$), indicating that on average, hippocampal phase angles led PFC phase angles, providing more evidence that hippocampal activity was leading PFC activity.
The phase differences agree with those displayed in \citealt[Figure 5.a]{brincat2015frequency} (though we examined a different time period of the trial and pooled only across significant PFC to hippocampus connections).
In contrast, Figure \ref{fig:phasediff}.B displays within-region phase differences tightly clustered around $0^\circ$ (mean: $-0.04$, 95\% CI: $[-0.2^\circ,0.2^\circ]$), indicating that within-region phase coupling may be driven mostly by spatial correlations in the recordings.
These results suggest that due to the CA3 to PFC and Sub to PFC connections identified in both the trivariate and 24-dimensional analyses and the overall negative phase difference (PFC - hippocampus), hippocampal activity leads PFC activity during this period of the task.
Importantly, while aggregated pairwise phase differences may have given some evidence of directionality, torus graphs provided detailed information about direct connections between hippocampal output subregions and PFC.

\begin{figure}
	\centering
	\includegraphics[width=0.7\linewidth]{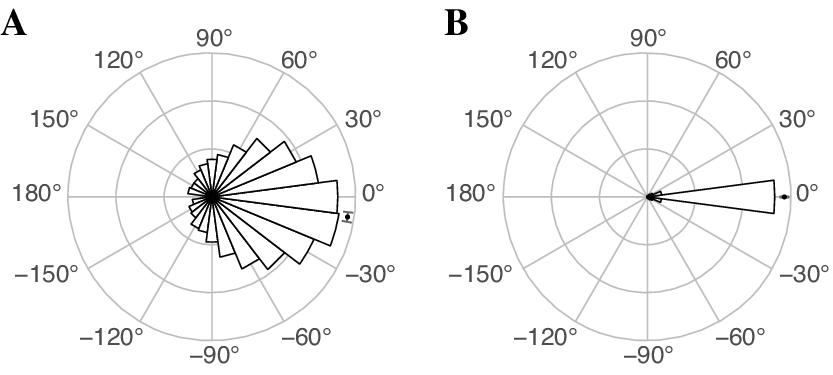}
	\caption{Circular histograms of phase differences, in degrees, from 24-dimensional LFP data for (A) significant connections between PFC and hippocampus (specifically, CA3 and Sub) and (B) significant PFC, CA3, and Sub within-region connections,  The mean phase offset with $95\%$ confidence interval is shown as a black dot with a narrow gray interval. Observations were pooled across all significant edges within or between the regions. The within-region phase differences were tightly concentrated around zero while the PFC-hippocampus phase differences were centered below zero, indicating a possible lead-lag relationship with hippocampus leading PFC.}
	\label{fig:phasediff}
\end{figure}

In the low-dimensional subnetworks, we found that torus graphs yielded intuitive results while PLV did not.
In particular, for the five-dimensional subnetwork consisting of electrodes along a linear probe, torus graphs inferred a nearest-neighbor conditional independence structure which we would expect for electrodes arranged linearly in space, while PLV inferred a fully connected graph.
In the trivariate subnetworks, torus graphs suggested connectivity between PFC and Sub and DG and Sub, but not between PFC and DG, while again, PLV inferred a fully connected graph.
Based on our analytic derivations in Section \ref{sec:tgphasecoupling} and simulation study in Section \ref{sec:simstudy}, PLV is likely not reflecting the correct dependence structures in either low-dimensional network, and is instead reflecting both direct and indirect connections between the nodes. 

When we investigated possible use of submodels, we found that a phase difference model appeared reasonable due to the lack of concentration in the phase sums, but that a uniform marginal model was not warranted.
While we used the phase difference model on the full 24-dimensional data set, we found that we would have obtained nearly the same results using the full torus graph model, with discrepancies for only a few edges in the post-hoc edgewise test results shown in Figure \ref{fig:AcrossROI}.B (which do not change the overall conclusions).
We found that torus graphs appeared to fit the data reasonably well, with both visual summaries and KS tests suggesting that the marginal distributions of each angle and of the pairwise sum and difference sufficient statistics were similar in the data and the fitted model.
In contrast, when we followed up by fitting a \textit{sine model} to the same data \citep[using code from][]{RodriguezLujan:2017eb}, we observed multimodality in the bivariate densities of the fitted model and a poor correspondence between the fitted and observed sufficient statistics, leading us to conclude that, as discussed in Section \ref{sec:tgsubmodels}, the \textit{sine model} fails to match the second-order dependence structure in the neural data due to low marginal concentration Figure \ref{fig:edawithsinemodel}.
KS tests comparing the fitted \textit{sine model} phase sums and differences to the data also suggested the data distribution does not match the \textit{sine model} distribution (KS/Fisher's method, $p < 0.0001$, for both phase sums and phase differences).

Torus graphs provided a good description of the neural phase angle data and provided substantive conclusions that could not have been obtained using bivariate phase coupling measures like PLV.

%% file: tex/discussion.tex
We have argued that torus graphs provide a natural analogue to Gaussian graphical models: 
Theorem \ref{thm:tgnatural} and Corollary \ref{cor:tgprop} 
show that starting with a full torus graph, which is an exponential family  with two-way interactions, setting a specific set of interaction coefficients $\bs{\phi}_{jk}$ to zero results in conditional independence of the $j$th and $k$th circular random variables.
We provided methods for fitting a torus graph to data, including identification of the graphical structure, i.e., finding the non-zero interaction coefficients, corresponding to edges in the graph (code, tutorial, and data provided at \url{https://github.com/natalieklein/torus-graphs}).
We also demonstrated 
that previous models in the literature amount to special cases, and therefore make additional assumptions that may or may not be appropriate for neural data.
In particular, while the uniform marginal model or the phase difference model may be reasonable for neural phase angle data, the most widely studied model in multivariate circular statistics, the \textit{sine model}, is less well-behaved and does not appear to be capable of matching the characteristics of neural phase angle data.
In addition, we showed that PLV is a measure of positive circular correlation under the assumption of uniform marginal distributions of the angles, but that PLV is unable to recover functional connectivity structure that takes account of multi-way dependence among the angles. 
In our analysis of LFP phases 300 ms after cue presentation in an associative memory task, the fitted torus graph correctly identified the apparent dependence structure of the linear probe within CA3; it suggested Sub may be responsible for apparent phase coupling between PFC and DG; and it led to the conclusion that, at this point in the task, hippocampus phases lead those from PFC (by $8.5^{\circ}$ with SE $=0.95^{\circ}$).

Here, our torus graphs were based on phases of oscillating signals, with no regard to their amplitudes.
This is different than phase amplitude coupling in which the phase of one oscillation may be related to the amplitude of an oscillation in a different frequency band \citep{tort2010measuring}.
Also, like other graph estimation methods, interpretations based on torus graphs assume that all relevant signals have been recorded, while in reality, they could be affected by unmeasured confounding variables (e.g., activity from other brain regions). 
In addition, in applications such as phase angles in LFP, several preprocessing steps (referencing, localization, and filtering) are needed to extract angles from the signals. 
The torus graph implementation we have described here ignores these steps, and takes well-defined angles as the starting point for analysis.
Furthermore, local field potentials tend to be highly spatially correlated, suggesting that inclusion of spatial information might be helpful for identifying structure.

Future work could include further investigation and theoretical analysis of how well torus graphs perform when the sample size is smaller relative to the dimension of the data.
In some data sets, the full torus graph with $2d^2$ parameters may be overparameterized, and estimation and inference may be more accurate using one of the subfamilies; we demonstrated a model selection approach that indicated that the phase difference submodel would be reasonable for this data set.
Furthermore, even when a full torus graph model is used, interpretability of the results could be enhanced by assessing evidence for reflectional and rotational dependence separately; that is, instead of putting a single edge based on the test of all four coupling parameters, we could construct a graph based only on the two parameters corresponding to rotational (or reflectional) dependence.
Finally, in the uniform marginal phase difference model, the strength of coupling for each type of dependence could be quantified using the measure we introduced in Section \ref{subsec:partialplv}, which falls between 0 and 1, facilitating comparison of relative strengths of the connections.

By extending existing models, torus graphs are able to represent a wide variety of multivariate circular data, including neural phase angle data. 
Extensions to this work could study changes in graph structure across time or across experimental conditions, and could investigate latent variable models involving hidden states, or a spatial hierarchy of effects. 
We anticipate a new line of research based on torus graphs.

%% file: texsupp/tgmodelcomplex.tex
To derive the appropriate first- and second-order sufficient statistics that correspond to first circular moments and to circular covariances, we first represent the angles as unit modulus complex random variables (where $i$ is the imaginary unit):
\begin{align*}
    Z_j = e^{i X_j}.
\end{align*}
The first circular moment is defined as
\begin{align*}
    E[Z_j] = R_j e^{i \mu_j}
\end{align*}
where $\mu_j$ is the mean direction and $R_j$ is the resultant length, so the corresponding complex first order sufficient statistic for a single observation is simply 
\begin{align*}
    z_j = \cos(x_j) + i \sin(x_j)
\end{align*}
which may be described as a real-valued sufficient statistic vector \begin{align*}
    \mbf{S}^1_j(x_j) = [\cos(x_j), \, \sin(x_j)]^T.
\end{align*} 

When considering second-order interactions between complex variables, there are two types of covariance \citep[Ch. 2.2]{schreier2010statistical}.
Rotational covariance between $Z_j$ and $Z_k$ is described by
\begin{align*}
    E[(e^{i X_j} - R_j e^{i \mu_j})\overline{(e^{i X_k} - R_k e^{i \mu_k})}] = E[e^{i(X_j-X_k)}] - R_j R_k e^{i (\mu_j-\mu_k)},
\end{align*}
where $\overline{Z_k}$ is the complex conjugate, while reflectional covariance is described by
\begin{align*}
    E[(e^{i X_j} - e^{i \mu_j})(e^{i X_k} - e^{i \mu_k})] = E[e^{i(X_j+X_k)}] - R_j R_k e^{i (\mu_j+\mu_k)}.
\end{align*}

This shows that in addition to the first-order statistics, we additionally need two more complex sufficient statistics to describe the second-order behavior:
\begin{align*}
    e^{i(x_j-x_k)} &= \cos(x_j-x_k) + i \sin(x_j-x_k) \\
    e^{i(x_j+x_k)} &= \cos(x_j+x_k) + i \sin(x_j+x_k).
\end{align*}
These may be collected into a real-valued vector 
\begin{align*}
     \mbf{S}^2_{jk}(x_j,x_k) = [\cos(x_i-x_j), \, \sin(x_i-x_j), \,
     \cos(x_i+x_j), \, \sin(x_i+x_j)]^T.
\end{align*}

Therefore, the canonical exponential family distribution given the first circular moments of each variable and the complete second-order interactions (rotational and reflectional) between each variable coincides with that given in Equation \ref{eq:tgdiff}:
\begin{align}   
p(\mbf{x}) &\propto
\exp\left\{ \sum_{j=1}^d \bs{\phi}_{j}^T \mbf{S}^1_j(x_j) + 
    \sum_{j<k} \bs{\phi}_{jk}^T 
    \mbf{S}^2_{jk}(x_j,x_k) \right\}\\
    &=
    \exp\left\{ \sum_{j=1}^d \bs{\phi}_{j}^T \begin{bmatrix}
    \cos (x_j) \\ \sin (x_j)
    \end{bmatrix} + 
    \sum_{j<k} \bs{\phi}_{jk}^T 
    \begin{bmatrix}
    \cos (x_j - x_k) \\ \sin (x_j - x_k) \\
    \cos (x_j + x_k) \\ \sin (x_j + x_k)
    \end{bmatrix} \right\}. \label{eq:tgmodel}
\end{align}
Thus, the torus graph model is maximum entropy with respect to constraints on the expected values of the sufficient statistics, that is, the circular first moments and complex covariances \cite{wainwright2008graphical}.
We note that, similar to the multivariate Gaussian distribution, the torus graph model only contains sufficient statistics for circular first moments and covariances, but it does not contain sufficient statistics corresponding to the second circular moment of a single angle $X_j$ (that is, it does not include interactions of the form $Z_j Z_j$); such a model was recently explored in \cite{navarro2017multivariate}.

The maximum entropy motivation for this model also offers some intuition for interpretation of the parameters; in particular, we see that the subvector $\bs{\phi}_{jk,1:2}$ corresponds to rotational covariance while the subvector $\bs{\phi}_{jk,3:4}$ corresponds to reflectional covariance.
However, the magnitude of each parameter is difficult to interpret directly because it depends not only on the covariance but also the marginal concentration of each variable (which is related to the resultant lengths $R_j$ and $R_k$) as well as the sum or difference of the mean directions.

%% file: texsupp/tgreparam.tex
While the canonical exponential family form in Equation  \ref{eq:tgmodel} is useful for understanding the maximum entropy constraints of the model and for deriving score matching estimators, it does not immediately appear similar to previous work in multivariate circular statistics (such as the \textit{sine model}).
To obtain another form that is easier to compare to previous work, we begin with an alternate parameterization that is similar to the \textit{sine model}, then show how it can be transformed into our parameterization.
Crucially, this transformation can also be reversed to potentially aid in interpretation of parameters.

Consider the \textit{mean-centered parameterization}
\begin{align*}
    p(\mbf{x} ; \bs{\theta}) \propto 
    \exp \left\{ 
        \sum_{j=1}^d \kappa_j \cos(x_j - \mu_j) + 
        \sum_{j < k} 
        \begin{bmatrix}
        \lambda^{cc}_{jk} \\ 
        \lambda^{cs}_{jk} \\
        \lambda^{sc}_{jk} \\
        \lambda^{ss}_{jk}
        \end{bmatrix}^T
        \begin{bmatrix}
        \cos(x_j - \mu_j)\cos(x_k - \mu_k) \\ 
        \cos(x_j - \mu_j)\sin(x_k - \mu_k) \\
        \sin(x_j - \mu_j)\cos(x_k - \mu_k) \\
        \sin(x_j - \mu_j)\sin(x_k - \mu_k)
        \end{bmatrix}
    \right\}
\end{align*}
where the parameters are $\bs{\theta} = \left[ \bs{\mu}, \bs{\kappa}, \bs{\lambda^{cc}}, \bs{\lambda^{cs}}, \bs{\lambda^{sc}}, \bs{\lambda^{ss}}\right]$ with the interpretation that $\mu_j \in [0, 2\pi)$ is the mean direction of $x_j$, $\kappa_j > 0$ is the marginal concentration of $x_j$, and the $\lambda$ parameters control interactions between angles.

In the univariate terms, we use trigonometric sum and difference formulas to rewrite
\begin{align*}
    \kappa_j \cos(x_j - \mu_j) = \kappa_j \cos(\mu_j) \cos(x_j) + \kappa_j \sin(\mu_j) \sin(x_j)
\end{align*}
so that in the parameterization of Equation \ref{eq:tgmodel}, 
\begin{align*}
    \bs{\phi}_j = 
    \begin{bmatrix} 
        \kappa_j \cos(\mu_j) \\
        \kappa_j \sin(\mu_j)
    \end{bmatrix}.
\end{align*}
Therefore, we can clearly calculate $\bs{\phi}_j$ for given $\kappa_j$ and $\mu_j$, and also, given $\bs{\phi}_j$, we have (using the Pythagorean theorem and definition of tangent)
\begin{align*}
    \kappa_j &= \sqrt{\bs{\phi}_{j,1}^2 + \bs{\phi}_{j,2}^2} \\
    \mu_j &= \arctan\left( \frac{\bs{\phi}_{j,2}}{\bs{\phi}_{j,1}}\right).
\end{align*}
Thus we have demonstrated a diffeomorphism between the two parameterizations for the marginal terms.

Similarly, for the pairwise coupling terms, we consider without loss of generality the pair $\{ X_j, X_k \}$ but for simplicity drop subscripts on the $\lambda$ parameters; using trigonometric sum and difference identities and simplifying, we find the pairwise coupling term is
\begin{align*}
    \frac{1}{2} 
    \begin{bmatrix}
        (\lambda^{cc}+\lambda^{ss})\cos(\mu_j-\mu_k) + (\lambda^{cs}-\lambda^{sc})\sin(\mu_j-\mu_k) \\
        (\lambda^{sc}-\lambda^{cs})\cos(\mu_j-\mu_k) + (\lambda^{cc}+\lambda^{ss})\sin(\mu_j-\mu_k) \\
        (\lambda^{cc}-\lambda^{ss})\cos(\mu_j+\mu_k) + (-\lambda^{cs}-\lambda^{sc})\sin(\mu_j+\mu_k) \\
        (\lambda^{cs}+\lambda^{sc})\cos(\mu_j+\mu_k) + (\lambda^{cc}-\lambda^{ss})\sin(\mu_j+\mu_k) 
    \end{bmatrix}^T
    \begin{bmatrix}
    \cos (x_j - x_k) \\ \sin (x_j - x_k) \\
    \cos (x_j + x_k) \\ \sin (x_j + x_k)
    \end{bmatrix} 
\end{align*}
so that it is straightforward to calculate $\bs{\phi}_{jk}$ given $\mu_j, \mu_k, \kappa_j, \kappa_k$, and the four $\lambda$ parameters.

The $\lambda$ parameters may be recovered as follows, where for brevity we use $\mu_{jk}^- = \mu_j - \mu_k$ and $\mu_{jk}^+ = \mu_j + \mu_k$:
\begin{align*}
\begin{split}
    \lambda^{cc} &= \bs{\phi}_{jk,1} \cos(\mu_{jk}^-) + 
                    \bs{\phi}_{jk,2} \sin(\mu_{jk}^-) + 
                    \bs{\phi}_{jk,3} \cos(\mu_{jk}^+) + 
                    \bs{\phi}_{jk,4} \sin(\mu_{jk}^+) \\
    \lambda^{cs} &= -\bs{\phi}_{jk,2} \cos(\mu_{jk}^-) + 
                    \bs{\phi}_{jk,1} \sin(\mu_{jk}^-) + 
                    \bs{\phi}_{jk,4} \cos(\mu_{jk}^+) - 
                    \bs{\phi}_{jk,3} \sin(\mu_{jk}^+) \\
    \lambda^{sc} &= \bs{\phi}_{jk,2} \cos(\mu_{jk}^-) - 
                    \bs{\phi}_{jk,1} \sin(\mu_{jk}^-) + 
                    \bs{\phi}_{jk,4} \cos(\mu_{jk}^+) - 
                    \bs{\phi}_{jk,3} \sin(\mu_{jk}^+) \\
    \lambda^{ss} &= \bs{\phi}_{jk,1} \cos(\mu_{jk}^-) + 
                    \bs{\phi}_{jk,2} \sin(\mu_{jk}^-) - 
                    \bs{\phi}_{jk,3} \cos(\mu_{jk}^+) - 
                    \bs{\phi}_{jk,4} \sin(\mu_{jk}^+). 
\end{split}
\end{align*}
This shows we have a diffeomorphism between the parameterizations.

Theorem 4.2.2 of \cite{kass2011geometrical} states that a subfamily of a regular exponential family is itself a (lower-dimensional) regular exponential family if and only if the subspace of the natural parameter space corresponding to the subfamily is an affine subspace of the natural parameter space. 
In the mean-centered parameterization, the \textit{sine model} has parameter constraints $\lambda^{cc} = \lambda^{cs} = \lambda^{sc} = 0$.
Given the equations above, the \textit{sine model} corresponds to a restriction of the natural parameter space of the torus graph density of Equation \ref{eq:tgmodel}:
\begin{align}
    \bs{\phi}_{jk} = \frac{1}{2} \lambda^{ss} [\cos(\mu_{jk}^-), \,
    \sin(\mu_{jk}^-), \, -\cos(\mu_{jk}^+), \,
    -\sin(\mu_{jk}^+)]^T.
\end{align}
This implies that the pairwise interactions must follow a specific structured form in the \textit{sine model}, where the magnitude of interactions is governed by $\lambda^{ss}$ and the following relationship between the parameters is observed (regardless of $\mu_j$ and $\mu_k$):
\begin{align}
    \bs{\phi}_{jk,1}^2 + \bs{\phi}_{jk,2}^2 = \bs{\phi}_{jk,3}^2 + \bs{\phi}_{jk,4}^2.
\end{align}
Because of the nonlinear relationship between the parameters, the subspace corresponding to the \textit{sine model} is not an affine subspace of the natural parameter space. Therefore, the \textit{sine model} is not itself a regular exponential family. On the other hand, the uniform marginal model and the phase difference model both are defined by setting components of the natural parameter to zero, as is the phase difference model with uniform marginals, so each of these families is itself a regular exponential family. This proves Theorem \ref{thm:tgsubmodels} in the main text.

%% file: texsupp/tgtheoremproof.tex
\begin{enumerate}
    \item Because exponential family models are maximum entropy models subject to constraints on the expected values of the sufficient statistics \citep{wainwright2008graphical}, the torus graph is the maximum entropy model subject to constraints on the first circular moments and complex covariances between angles (following the derivation in Section \ref{sec:tgmodelcomplex} that relates the sufficient statistics to circular first moments and to complex covariances).
    \item The torus graph density is positive and continuous on ${[0,2\pi)}^d$ and factorizes into pairwise interaction terms as shown in Equation \ref{eq:tgmodel}.
By the Hammersley-Clifford theorem \citep{lauritzen1996graphical}, the random variables $X_j$ and $X_k$ are conditionally independent given all other variables if and only if $\bs{\phi}_{jk} = \mbf{0}$.
\end{enumerate}

%% file: texsupp/phasediff.tex
First, we state the Harmonic Addition Theorem which will be very useful throughout this set of derivations (see \cite{weisstein} for proof).
\begin{thm}[Harmonic Addition Theorem] \label{thm:harmonicaddition}
	The weighted sum of cosine functions with the same period and arbitrary phase shifts is the cosine function
	\begin{align*}
		\sum_{i=1}^n a_i \cos(x-\delta_i) = A \cos(x-\Delta)
	\end{align*}
	where  
	\begin{align*}
		b_x &= \sum_{i=1}^n a_i \cos(\delta_i) \\
		b_y& = \sum_{i=1}^n a_i \sin(\delta_i) \\
		A &= \sqrt{b_x^2 +b_y^2} \\
		\Delta &= \arctan\left(\frac{b_y}{b_x}\right).
	\end{align*}
\end{thm}
Throughout these derivations, when we use the arctangent function $\arctan(\cdot)$, it is understood that the angular value is chosen to fall in the same interval as the random variables (in this case, $[0, 2\pi)$, though other intervals such as $[-\pi,\pi)$ could be chosen). 
We will use the notation $\bs{\phi}_{jk} = [\alpha_{jk}, \beta_{jk}, \gamma_{jk}, \delta_{jk}]^T$ to refer to elements of the pairwise coupling parameter vector.

For the bivariate torus graph model, we derive the distribution of phase differences to compare with bivariate phase coupling measures, which depend on the distribution of phase differences.
Let $\theta = X_1 - X_2$ be a random variable with support $(-2\pi, 2\pi)$ (as $X_1 \in [0, 2\pi)$ and $X_2 \in [0, 2\pi)$)  and let $p_{X_1,X_2}(x_1,x_2)$ denote the bivariate phase difference model density. 
Applying the change of variables $\theta = X_1-X_2$ and trigonometric identity $\sin(\theta)=\cos\left(\theta-\tfrac{\pi}{2}\right)$ yields
\begin{align*}
	\begin{split}
		p_{\theta,X_2}(\theta,x_2) &= p_{X_1,X_2}(\theta+x_2,x_2) \\ &\propto \exp\left\{ \kappa_1 \cos(x_2-(\mu_1-\theta)) + \kappa_2 \cos(x_2-\mu_2)  \right\} \times \\
		&\hspace{2em}\exp\left\{ \alpha_{12}\cos(\theta)+ \beta_{12}\cos\left(\theta-\tfrac{\pi}{2}\right)\right\}.
	\end{split}
\end{align*}
Applying Theorem \ref{thm:harmonicaddition} to each factor,  
\begin{align*}
	p_{X_1,X_2}(\theta+x_2,x_2) \propto  \exp\left\{ A_1\cos(\theta-\Delta_1)\right\} \exp\left\{ A_2\cos(x_2-\Delta_2) \right\}
\end{align*}
where 
\begin{align*}
	A_1 &= \sqrt{\alpha_{12}^2+\beta_{12}^2}\\
	\Delta_1 &= \arctan\left(\tfrac{\beta_{12}}{\alpha_{12}}\right) \\
	A_2(\theta) &= \sqrt{\kappa_1^2+\kappa_2^2+2\kappa_1\kappa_2\cos(\theta-(\mu_1-\mu_2))} \\
	\Delta_2(\theta) &= \arctan\left(\tfrac{\kappa_1\sin(\mu_1-\theta) + \kappa_2\sin(\mu_2)}{\kappa_1\cos(\mu_1-\theta) + \kappa_2\cos(\mu_2)}\right),
\end{align*}
and we use the notation $A_2(\theta), \Delta_2(\theta)$ to indicate that these are functions of $\theta$.

To obtain the marginal density of $\theta$, we need to integrate over $x_2$ and also wrap the resulting distribution back to the support $[0, 2\pi)$ (though we could choose a different support of length $2\pi$, such as $[-\pi,\pi)$, if desired). 
Notice that $p_{X_1,X_2}(\theta+x_2,x_2)$ has constraints $X_1, X_2 \in [0,2\pi)$, implying that $0 \le \theta+X_2 < 2\pi$ so $-\theta \le X_2 < 2\pi-\theta$.
This means that when $\theta < 0$, $X_2 \in [-\theta, 2\pi)$ and when $\theta > 0$, $X_2 \in [0,2\pi-\theta)$, so the marginal distribution of $\theta$ is defined piecewise:
\begin{align*}
	p_\theta(\theta) \propto \left\{  
	\begin{array}{ll} 
		\mathbbm{1}_{(\theta \in [-2\pi, 0])} g(\theta) \int_{-\theta}^{2\pi} \exp\left\{ A_2(\theta)\cos(x_2-\Delta_2(\theta)) \right\} \, dx_2  \\
		\mathbbm{1}_{(\theta \in [0, 2\pi])} g(\theta) \int_{0}^{2\pi-\theta} \exp\left\{ A_2(\theta)\cos(x_2-\Delta_2(\theta)) \right\} \, dx_2
	\end{array} \right.
\end{align*}
where $g(\theta) =  \exp\left\{ A_1\cos(\theta-\Delta_1)\right\}$. 

Define $W = \theta \, (\text{mod} \, 2\pi)$ to be the wrapped version of $\theta$ so that $W \in [0,2\pi)$ has the wrapped distribution 
\begin{align*}
	p_W(w) &= p_\theta(w) + p_\theta(w-2\pi) \\
	&\propto g(w) \int_{0}^{2\pi-w} \exp\left\{ A_2(w) \cos(x_2-\Delta_2(w) ) \right\} \, dx_2 \\
	&\hspace{1em}+ g(w-2\pi) \int_{2\pi-w}^{2\pi} \exp\left\{ A_2(w-2\pi)\cos(x_2-\Delta_2(w-2\pi)) \right\} \, dx_2
\end{align*}
Using the fact that $g$, $A_2$, and $\Delta_2$ are $2\pi$-periodic functions and the definition of $I_0$ (the modified Bessel function of the first kind), we obtain
\begin{align*}
	p_W(w) &\propto g(w) \int_{0}^{2\pi} \exp\left\{ A_2(w)\cos(x_2-\Delta_2(w)) \right\} \, dx_2 = g(w) I_0(A_2(w)).
\end{align*}
Thus, $g(w)$ is the direct coupling term and $f(w) = I_0(A_2(w))$ is the marginal concentration term.

The derivation for phase differences from the trivariate torus graph model uses similar techniques.
Consider a trivariate torus graph with $\kappa_1=\kappa_2=\kappa_3=0$ for simplicity; applying trigonometric identities yields
\begin{align*}
	\begin{split}
	p(x_1,x_2,x_3)&\propto \exp \left\{ \sum_{(i,j)\in E} 
	\begin{bmatrix} 
	    \alpha_{ij} \\ 
	    \beta_{ij} \\ 
	    \gamma_{ij} \\ 
	    \delta_{ij}
	\end{bmatrix}^T 
	\begin{bmatrix} 
	    \cos (x_i - x_j) \\  
	    \sin (x_i - x_j) \\ 
	    \cos (x_i + x_j) \\ 
	    \sin (x_i + x_j)
	\end{bmatrix} \right\} \\
	&=\exp \left\{ \sum_{(i,j)\in E} \begin{bmatrix} 
	    \alpha_{ij} \\ 
	    \beta_{ij} \\ 
	    \gamma_{ij} \\
	    \delta_{ij}
	\end{bmatrix}^T 
	\begin{bmatrix} 
	    \cos (x_i - x_j-0) \\  
	    \cos (x_i - x_j-\pi/2) \\ 
	    \cos (x_i + x_j-0) \\ 
	    \cos (x_i + x_j-\pi/2)
	\end{bmatrix} \right\}	
	\end{split}
\end{align*}
where $E=\{ (1,2),\ (1,3), \ (2,3) \}$. 
Apply the variable transformation $\theta_{12}=x_1-x_2$ and expand the expression above:
\begin{align*}
	\begin{split}
	p(\theta_{12},x_2,x_3)&\propto \exp \left\{ 
	\begin{bmatrix} 
	    \alpha_{12} \\ 
	    \beta_{12} \\ 
	    \gamma_{12} \\ \delta_{12}
	\end{bmatrix}^T 
	\begin{bmatrix} 
	    \cos (\theta_{12}-0) \\  
	    \cos (\theta_{12}-\pi/2) \\ 
	    \cos (\theta_{12}+2x_2-0) \\ 
	    \cos (\theta_{12}+2x_2-\pi/2)
	\end{bmatrix} \right\} \times \\
	&\hspace{2em}\exp \left\{ 
	\begin{bmatrix} 
	    \alpha_{13} \\ 
	    \beta_{13} \\ 
	    \gamma_{13} \\ 
	    \delta_{13}
	\end{bmatrix}^T 
	\begin{bmatrix} 
	    \cos (x_1-x_3-0) \\  
	    \cos (x_1-x_3-\pi/2) \\ 
	    \cos (x_1+x_3-0) \\ 
	    \cos (x_1+x_3-\pi/2)
	\end{bmatrix} \right\} \times \\
	&\hspace{2em}\exp \left\{ 
	\begin{bmatrix} 
	    \alpha_{23} \\ 
	    \beta_{23} \\ 
	    \gamma_{23} \\ 
	    \delta_{23}
	\end{bmatrix}^T 
	\begin{bmatrix} 
	    \cos (x_2-x_3-0) \\  
	    \cos (x_2-x_3-\pi/2) \\ 
	    \cos (x_2+x_3-0) \\ 
	    \cos (x_2+x_3-\pi/2)
	\end{bmatrix} \right\}\\
	\end{split}			       
\end{align*}

To get the marginal distribution of $\theta_{12}$, we need to integrate out other variables, which is not tractable analytically for the full torus graph model, so we consider the phase difference model which corresponds to setting $\gamma = \delta = 0$ for all pairs. 
This has the effect of making the density depend only on phase differences.
\begin{align*}
	\begin{split}
		p(\theta_{12},x_2,x_3)&\propto \exp \left\{ 
		\begin{bmatrix} 
		    \alpha_{12} \\ 
		    \beta_{12} 
		\end{bmatrix}^T 
		\begin{bmatrix} 
		    \cos (\theta_{12}-0) \\  
		    \cos (\theta_{12}-\pi/2)
		\end{bmatrix} \right\} \times \\
		&\hspace{2em}\exp \left\{ 
		\begin{bmatrix} 
		    \alpha_{13} \\ 
		    \beta_{13} 
		\end{bmatrix}^T 
        \begin{bmatrix} 
            \cos (x_1-x_3-0) \\  
            \cos (x_1-x_3-\pi/2) 
        \end{bmatrix} \right\} \times \\
        &\hspace{2em}\exp \left\{ 
        \begin{bmatrix} 
            \alpha_{23} \\ 
            \beta_{23} 
        \end{bmatrix}^T 
		\begin{bmatrix} 
		    \cos (x_2-x_3-0) \\  
		    \cos (x_2-x_3-\pi/2)
		\end{bmatrix} \right\}.
	\end{split}
\end{align*}
Similarly to the bivariate case, we apply Theorem \ref{thm:harmonicaddition} to each factor; the first factor, $g(\theta_{12}) =  \exp\left\{ A_1\cos(\theta_{12}-\Delta_1)\right\}$, has the same form as in the bivariate case and represents the direct coupling between $X_1$ and $X_2$.
The second factor may also be written as
\begin{align*}
	\exp \left\{ A_{13}\cos(x_1-x_3- \Delta_{13}\right\}
\end{align*}
where $A_{13} = \sqrt{\alpha_{13}^2 + \beta_{13}^2}$ and $\Delta_{13} = \arctan(\beta_{13}/\alpha_{13})$,
and the third factor may also be written as
\begin{align*}
	\exp \left\{ A_{23}\cos(x_2-x_3- \Delta_{23}\right\},
\end{align*}
where $A_{23} = \sqrt{\alpha_{23}^2 + \beta_{23}^2}$ and $\Delta_{23} = \arctan(\beta_{23}/\alpha_{23})$.  
Next we apply Theorem \ref{thm:harmonicaddition} again to combine the second and third factors:
\begin{align*}
	p(\theta_{12},x_3)\propto \exp\left\{ A_1\cos(\theta_{12}-\Delta_1)\right\} \times \exp \left\{ A_{3}(\theta_{12})\cos(x_3- \Delta_{3}) \right\},
\end{align*}
where $A_3(\theta_{12})$ is
\begin{align*}
	A_3(\theta_{12}) = &\Bigg[\Big(   A_{13}\cos(x_1-\Delta_{13}) + A_{23}\cos(x_2-\Delta_{23}) \Big)^2 \\
	&\hspace{1em}+ \Big( A_{13}\sin(x_1-\Delta_{13}) + A_{23}\sin(x_2-\Delta_{23})  \Big)^2 \Bigg]^{1/2}.
\end{align*}
Because we will integrate out $x_3$, the form of $\Delta_3$ is not important (it will not affect the integral on this circular domain). 

Expanding the squares, simplifying, and using a trigonometric sum identity yields
\begin{align}
	A_3(\theta_{12}) = \sqrt{\alpha_{13}^2+\beta_{13}^2+\alpha_{23}^2+\beta_{23}^2+2t\cos(\theta_{12}-u)} \label{eq:a2triv}
\end{align}
where 
\begin{align*}
	t &= \sqrt{(\alpha_{13}^2+\beta_{13}^2)(\alpha_{23}^2+\beta_{23}^2)} \\
	u &= \Delta_{13}-\Delta_{23} = \arctan\left(\frac{\beta_{13}}{\alpha_{13}}\right) - \arctan\left(\frac{\beta_{23}}{\alpha_{23}}\right).
\end{align*}

Similar to the bivariate case, we integrate over $x_3$ and wrap the resulting distribution to obtain the marginal distribution of the wrapped phase difference $W = \theta_{12} \, (\text{mod} \, 2\pi)$:
\begin{align*}
	\begin{split}
	p(w) &\propto \exp \left\{ A_1 \cos \left( w - \Delta_1 \right) \right\} \int_0^{2\pi} \exp \left\{ A_3(w) \cos \left( x_3 - \Delta_3 \right) \right\} \, dx_3 \\
	&\propto \exp \left\{ A_1 \cos \left( w - \Delta_1\right) \right\} I_0(A_3(w) )
	\end{split}
\end{align*}
where $A_1 = \sqrt{\alpha_{12}^2 + \beta_{12}^2}$, $\Delta_1 = \arctan(\beta_{12}/\alpha_{12})$, and $A_3(w)$ is given in Equation \ref{eq:a2triv}.
Thus, $g(w)$ is the direct coupling term, and $h(w)$ is the indirect coupling term.

%% file: texsupp/scorematch.tex
Let $p_\mbf{X}(\mbf{x})$ be the unknown $d$-dimensional circular data density and $p(\mbf{x};\bs{\phi}) = \frac{1}{Z(\bs{\phi})} q(\mbf{x};\bs{\phi})$ be a $d$-dimensional model density with parameter vector $\bs{\phi} \in \mathbb{R}^m$.
Define the log model density gradient  $\bs{\psi}:[0,2\pi)^d \rightarrow \mathbb{R}^d$ as $\bs{\psi}(\mbf{x};\bs{\phi}) = \nabla_\mbf{x} \log q(\mbf{x};\bs{\phi})$; similarly, let $\bs{\psi}_\mbf{X}(\mbf{x}) = \nabla_\mbf{x} \log p_\mbf{X}(\mbf{x})$.

To prove Theorem \ref{thm:scorematch} in the main text, make the following regularity assumptions:
\begin{enumerate}[label=\Alph*.]
	\item For all $i \in \left\{1,...,d\right\}$, $\bs{\psi}(\mbf{x};\bs{\phi})$ is differentiable with respect to $\mbf{x}_i$ on $[0,2\pi)$.
	\item For all $\bs{\phi}$, $E_\mbf{x} \left[ || \bs{\psi}(\mbf{x};\bs{\phi}) ||^2\right]$ and $E_\mbf{x} \left[ || \bs{\psi}_\mbf{X}(\mbf{x})  ||^2\right]$ are finite.
\end{enumerate}
These assumptions clearly hold for torus graphs as the log density is comprised of finite linear combinations of sine and cosine functions of $\mbf{x}$, each of which is infinitely differentiable with derivatives bounded within $[-1, 1]$.
Note that we need one less assumption than the original formulation of score matching in \cite{Hyvarinen:2005wb} due to the circular nature of the density.
\begin{proof}[Proof of Theorem \ref{thm:scorematch}]

First, we show that the score matching objective function only depends on the unknown data density through an expectation.

Expanding the squared difference gives
\begin{align*}
	\begin{split}
	J(\bs{\phi}) &= 
	\int_0^{2\pi} p_\mbf{X}(\mbf{x}) \left[ \tfrac{1}{2} || \nabla_\mbf{x} \log p_\mbf{X}(\mbf{x}) ||_2^2 \right] \; d\mbf{x}  \\
	&\hspace{2em}+\int_0^{2\pi} p_\mbf{X}(\mbf{x}) \left[ \tfrac{1}{2} || \nabla_\mbf{x} \log q(\mbf{x};\bs{\phi}) ||_2^2  \right] \; d\mbf{x}  \\
	&\hspace{2em}-\int_0^{2\pi}  p_\mbf{X}(\mbf{x}) [\nabla_\mbf{x} \log q(\mbf{x};\bs{\phi})]^T [\nabla_\mbf{x} \log p_\mbf{X}(\mbf{x})]\, d\mbf{x}.
	\end{split}
\end{align*}
The first term does not depend on $\bs{\phi}$ and the second term is already in terms of an expectation over the data density, so we focus now on the third term (call it $A$):
\begin{align*}
	A &= - \int_0^{2\pi} p_\mbf{X}(\mbf{x}) \left[ \sum_{i=1}^d \bs{\psi}_i(\mbf{x};\bs{\phi})\bs{\psi}_{\mbf{X},i}(\mbf{x}) \right] \; d\mbf{x} \\
	&= - \sum_{i=1}^d \int_0^{2\pi} \left[ \int_0^{2\pi} p_\mbf{X}(\mbf{x}) \bs{\psi}_i(\mbf{x};\bs{\phi})\bs{\psi}_{\mbf{X},i}(\mbf{x}) \; d\mbf{x}_i \right] \; d\mbf{x}_{-i} \\
	&= - \sum_{i=1}^d \int_0^{2\pi}  \left[ \int_0^{2\pi}  p_\mbf{X}(\mbf{x}) \frac{\partial}{\partial \mbf{x}_i} \log p_\mbf{X}(\mbf{x}) \bs{\psi}_i(\mbf{x};\bs{\phi}) \; d\mbf{x}_i \right] \; d\mbf{x}_{-i} \\
	&= - \sum_{i=1}^d \int_0^{2\pi}  \left[ \int_0^{2\pi}  \frac{\partial}{\partial \mbf{x}_i} p_\mbf{X}(\mbf{x}) \bs{\psi}_i(\mbf{x};\bs{\phi}) \; d\mbf{x}_i \right] \; d\mbf{x}_{-i}.
\end{align*}
Applying integration by parts, the inner integral becomes 
\begin{align*}
	p_\mbf{X}(\mbf{x}) \bs{\psi}_i(\mbf{x};\bs{\phi}) \Big|^{\mbf{x}_i=2\pi} _{\mbf{x}_i=0} - \int_0^{2\pi} p_\mbf{X}(\mbf{x}) \frac{\partial}{\partial \mbf{x}_i} \left(\bs{\psi}_i(\mbf{x};\bs{\phi}) \right) \; d\mbf{x}_i.
\end{align*}
Notice that because the variables are circular on $[0, 2\pi)$, 
\begin{align*}
	p_\mbf{X}(\mbf{x}) \big|_{\mbf{x}_i=0} &= p_\mbf{X}(\mbf{x}) \big|_{\mbf{x}_i=2\pi} \\
	\bs{\psi}_i(\mbf{x};\bs{\phi}) \big|_{\mbf{x}_i=0} &= \bs{\psi}_i(\mbf{x};\bs{\phi}) \big|_{\mbf{x}_i=2\pi}
\end{align*}
Therefore, $p_\mbf{X}(\mbf{x}) \bs{\psi}_i(\mbf{x};\bs{\phi}) \big|^{\mbf{x}_i=2\pi} _{\mbf{x}_i=0} = 0$, so $A$ becomes
\begin{align*}
	A &= - \sum_{i=1}^d \int_0^{2\pi}  \left[ - \int_0^{2\pi} p_\mbf{X}(\mbf{x}) \frac{\partial}{\partial \mbf{x}_i} \left(\bs{\psi}_i(\mbf{x};\bs{\phi}) \right) \; d\mbf{x}_i \right] \; d\mbf{x}_{-i} \\
	&= \int_0^{2\pi} p_\mbf{X}(\mbf{x}) \left[ \sum_{i=1}^d \frac{\partial}{\partial \mbf{x}_i} \left(\bs{\psi}_i(\mbf{x};\bs{\phi}) \right) \right] \; d\mbf{x}
\end{align*}
Therefore, the score matching objective is
\begin{align}
	\begin{split}
 	J(\bs{\phi}) &= C + \int_0^{2\pi} p_\mbf{X}(\mbf{x}) \left[ \tfrac{1}{2} || \bs{\psi}(\mbf{x};\bs{\phi}) ||^2 \right] \; d\mbf{x} \\
 	&\hspace{2em}+\int_0^{2\pi} p_\mbf{X}(\mbf{x}) \left[ \sum_{i=1}^d \frac{\partial}{\partial \mbf{x}_i} \left(\bs{\psi}_i(\mbf{x};\bs{\phi}) \right) \right] \; d\mbf{x} \\
	&= C + E_{\mbf{x}} \left\{ \tfrac{1}{2} || \bs{\psi}(\mbf{x};\bs{\phi}) ||^2 + \sum_{i=1}^d \frac{\partial}{\partial \mbf{x}_i} \left(\bs{\psi}_i(\mbf{x};\bs{\phi}) \right)\right\}
	\end{split} \label{eq:smobj}
\end{align}
where $C$ does not depend on $\bs{\phi}$ and may be ignored without affecting the minima of the objective function.
This coincides with the form of score matching given in \cite{Hyvarinen:2005wb} except with the integral over the circular domain $[0, 2\pi)$.

Next, we show the explicit form of the score matching estimator for torus graphs.
As shown in \cite{forbes2015linear,yu2018graphical}, for exponential families, this score matching estimator is quadratic in the parameters.
Specifically, the torus graph density in Theorem \ref{thm:tgnatural} has a log density of the form
\begin{align*}
	\log q(\mbf{x};\bs{\phi}) = \bs{\phi}^T \mbf{S}(\mbf{x})
\end{align*}
where $\bs{\phi}$ are vectors of length $m = 2 d^2$ (the number of sufficient statistics). 

Therefore,
\begin{align*}
	\bs{\psi}(\mbf{x};\bs{\phi}) = \bs{\phi}^T \mbf{D}(\mbf{x})
\end{align*}
where the Jacobian $\mbf{D}(\mbf{x})$ is $m \times d$ with $i,j$th element $\frac{\partial}{\partial \mbf{x}_j} \mbf{S}_i$. 
Thus the first term inside the expectation in the score matching objective of Equation \ref{eq:smobj} may be written
\begin{align*}
	\tfrac{1}{2} || \bs{\psi}(\mbf{x};\bs{\phi}) ||^2 = \tfrac{1}{2} \bs{\phi}^T \mbf{D}(\mbf{x}) \mbf{D}(\mbf{x})^T \bs{\phi} \equiv \tfrac{1}{2} \bs{\phi}^T \bs{\Gamma}(\mbf{x}) \bs{\phi}.
\end{align*}
The elements of $\mbf{D}(\mbf{x})$ correspond to partial derivatives of the sufficient statistics with respect to the data.
The derivatives of the univariate sufficient statistics $\mbf{S}^1$ are given by
\begin{align*}
    \frac{\partial}{\partial x_\ell} \cos(x_j) &= \left\{ \begin{array}{rr} 
        -\sin(x_j), & \ell = j \\
        0, & \ell \neq j  
    \end{array}\right. , \\
    \frac{\partial}{\partial x_\ell} \sin(x_j) &= \left\{ \begin{array}{rr} 
        \hphantom{-}\cos(x_j), & \ell = j \\
        0, & \ell \neq j  
    \end{array}\right. .
\end{align*}
Similarly, the derivatives of the pairwise sufficient statistics $\mbf{S}^2$ may be calculated as
\begin{align*}
    \frac{\partial}{\partial x_\ell} \cos(x_j-x_k) &= \left\{ \begin{array}{rl} 
        -\sin(x_j-x_k), & \; \ell = j \\
        \sin(x_j-x_k), & \; \ell = k \\
        0, & \;  \ell \not\in \{j,k\}  
    \end{array}\right. , \\
    \frac{\partial}{\partial x_\ell} \sin(x_j-x_k) &= \left\{ \begin{array}{rl} 
        \cos(x_j-x_k), & \; \ell = j \\
        -\cos(x_j-x_k), & \; \ell = k \\
        0, & \; \ell \not\in \{j,k\} 
    \end{array}\right. , \\
    \frac{\partial}{\partial x_\ell} \cos(x_j+x_k) &= \left\{ \begin{array}{rl} 
        -\sin(x_j+x_k), & \; \ell \in \{j,k\} \\
        0, & \; \ell \not\in \{j,k\}   
    \end{array}\right. , \\
    \frac{\partial}{\partial x_\ell} \sin(x_j+x_k) &= \left\{ \begin{array}{rl} 
        \hphantom{-}\cos(x_j+x_k), & \; \ell \in \{j,k\} \\
        0, & \; \ell \not\in \{j,k\}  
    \end{array}\right. .
\end{align*}

Now we show that the second term inside the expectation in Equation \ref{eq:smobj} may be written simply in terms of the sufficient statitsics.
Notice that the $i$th element of the gradient may be written in terms of columns of the Jacobian:
\begin{align*}
	\bs{\psi}_i(\mbf{x};\bs{\phi}) = \bs{\phi}^T [\mbf{D}(\mbf{x})]_{\cdot,i}
\end{align*}
so that
\begin{align*}
	\frac{\partial}{\partial \mbf{x}_i} \left(\bs{\psi}_i(\mbf{x};\bs{\phi}) \right) = \bs{\phi}^T \frac{\partial}{\partial \mbf{x}_i} [\mbf{D}(\mbf{x})]_{\cdot,i}.
\end{align*}
Therefore, the second term inside the expectation in the score matching objective may be written
\begin{align*}
	\sum_{i=1}^d \frac{\partial}{\partial \mbf{x}_i} \left(\bs{\psi}_i(\mbf{x};\bs{\phi}) \right) = \bs{\phi}^T \left[ \sum_{i=1}^d \frac{\partial}{\partial \mbf{x}_i} [\mbf{D}(\mbf{x})]_{\cdot,i} \right] \equiv \bs{\phi}^T \mbf{H}(\mbf{x})
\end{align*}
where
\begin{align*}
    \mbf{H}(\mbf{x}) = [\mbf{S}^1(\mbf{x}), \,
 2\mbf{S}^2(\mbf{x})]^T.
\end{align*}
This relation holds because all nonzero elements of $\mbf{D}(\mbf{x})$ come from derivatives of sines and cosines; due to the relations $\frac{d}{dx} \cos(x) = -\sin(x)$ and $\frac{d}{dx} \sin(x) = \cos(x)$, taking the derivative again essentially converts the elements back to sufficient statistics.
\end{proof}

%% file: texsupp/conditional.tex
We prove that the distribution of one angle conditional on the other angles is von Mises as stated in Theorem \ref{thm:conditional}, enabling the use of Gibbs sampling for drawing samples from the distribution. 
We will use the notation $\bs{\phi}_{jk} = [\alpha_{jk}, \beta_{jk}, \gamma_{jk}, \delta_{jk}]^T$ to refer to elements of the pairwise coupling parameter vector.
\begin{proof}
	Let $c^-_{ij} = \cos(x_i-x_j)$, $c^+_{ij} = \cos(x_i+x_j)$, $s^-_{ij} = \sin(x_i-x_j)$, and $s^+_{ij} = \sin(x_i+x_j)$.
	Factor the torus graph density into terms containing $X_k$ and not containing $X_k$:
	\begin{align}
		\begin{split}
		p(\mbf{x}; \bs{\phi}) = 
			&C(\bs{\phi}) \exp\left\{ \sum_{i \neq k} \kappa_i \cos(x_i - \mu_i) + 
				\sum_{i < j, j \neq k} \begin{bmatrix} \alpha_{ij} \\ \beta_{ij} \\ \gamma_{ij} \\ \delta_{ij}\end{bmatrix}^T 
				\begin{bmatrix} c^-_{ij} \\  s^-_{ij} \\ c^+_{ij} \\ s^+_{ij}\end{bmatrix} \right\} \times \\
				&\hspace{2em}\exp\left\{ \kappa_k \cos(x_k - \mu_k) \right\} \times \\
				&\hspace{2em}\exp \left\{ \sum_{i<k} \begin{bmatrix} \alpha_{ik} \\ \beta_{ik} \\ \gamma_{ik} \\ \delta_{ik}\end{bmatrix}^T 
				\begin{bmatrix} c^-_{ik} \\  s^-_{ik} \\ c^+_{ik} \\ s^+_{ik} \end{bmatrix} 
				+ \sum_{i>k} \begin{bmatrix} \alpha_{ki} \\ \beta_{ki} \\ \gamma_{ki} \\ \delta_{ki}\end{bmatrix}^T 
				\begin{bmatrix} c^-_{ki} \\  s^-_{ik} \\ c^+_{ik} \\ s^+_{ik} \end{bmatrix}\right\}. \label{eq:factdens}
		\end{split}
	\end{align}
	Let the first factor (including the normalization constant) be denoted by $g(\mbf{x}_{-d}; \bs{\phi})$ and the second and third factors be denoted by $f(\mbf{x}; \bs{\phi})$.
	Then the conditional distribution is
	\begin{align*}
		p(x_k | \mbf{x}_{-k}; \bs{\phi}) = \frac{g(\mbf{x}_{-k}; \bs{\phi})f(\mbf{x}; \bs{\phi})}{g(\mbf{x}_{-k}; \bs{\phi}) \int f(\mbf{x}; \bs{\phi}) \, dx_k} = \frac{f(\mbf{x}; \bs{\phi})}{\int f(\mbf{x}; \bs{\phi}) \, dx_k}.
	\end{align*}
	Applying trigonometric identities to the third factor of Equation \ref{eq:factdens} and simplifying, we have
	\begin{align*}
		&\exp \left\{ \sum_{i<k} \begin{bmatrix} \alpha_{ik} \\ \beta_{ik} \\ \gamma_{ik} \\ \delta_{ik}\end{bmatrix}^T 
				\begin{bmatrix} c^-_{ik} \\  s^-_{ik} \\ c^+_{ik} \\ s^+_{ik} \end{bmatrix} 
				+ \sum_{i>k} \begin{bmatrix} \alpha_{ki} \\ \beta_{ki} \\ \gamma_{ki} \\ \delta_{ki}\end{bmatrix}^T 
				\begin{bmatrix} c^-_{ki} \\  s^-_{ik} \\ c^+_{ik} \\ s^+_{ik} \end{bmatrix} \right\} \\
				&= \exp \left\{ \sum_{i<k} \begin{bmatrix} \alpha_{ik} \\ \beta_{ik} \\ \gamma_{ik} \\ \delta_{ik}\end{bmatrix}^T 
				\begin{bmatrix} \cos(x_k-x_i) \\  \cos(x_k - x_i + \pi/2) \\ \cos(x_k + x_i) \\ \cos(x_k + x_i - \pi/2) \end{bmatrix} \right\} \times \\
				&\hspace{2em}\exp \left\{ \sum_{i>k} \begin{bmatrix} \alpha_{ki} \\ \beta_{ki} \\ \gamma_{ki} \\ \delta_{ki}\end{bmatrix}^T 
				\begin{bmatrix}  \cos(x_k-x_i)  \\  \cos(x_k - x_i - \pi/2) \\ \cos(x_k + x_i) \\ \cos(x_k + x_i - \pi/2) \end{bmatrix} \right\} \\
		&= \exp \left\{ \sum_{i \neq k} \begin{bmatrix} \alpha_{ik} \\ \beta_{ik} \\ \gamma_{ik} \\ \delta_{ik}\end{bmatrix}^T 
				\begin{bmatrix} \cos(x_k-x_i) \\  \cos(x_k - x_i +\sgn(i-k)\pi/2) \\ \cos(x_k + x_i) \\ \cos(x_k + x_i - \pi/2) \end{bmatrix} \right\}
	\end{align*}
	where, with slight abuse of notation, we let, for instance, $\alpha_{ik}$ denote either $\alpha_{ik}$ if $i < k$ or $\alpha_{ki}$ if $i > k$, and $\sgn(\cdot)$ is the signum function.
	Now we see $f(\mbf{x}; \bs{\phi})$ is a sum of cosine functions with argument $x_k$, so applying Theorem \ref{thm:harmonicaddition}, we have 
	\begin{align*}
		f(\mbf{x}; \bs{\phi}) = \exp(A \cos(x_k - \Delta))
	\end{align*}
	where $A = \sqrt{b_x^2 + b_y^2}$, $\Delta = \arctan\left(b_y/b_x\right)$, defined as
	\begin{align*}
	 	b_x &= \sum_m L_m \cos(V_m) \\
		b_y &= \sum_m L_m \sin(V_m) \\
		L &= \left[ \kappa_k, \bs{\alpha}_{\cdot,k}, \bs{\beta}_{\cdot,k}, \bs{\gamma}_{\cdot,k}, \bs{\delta}_{\cdot,k} \right] = [ \kappa_k, \bs{\phi}_{\cdot k}] \\
		V &= [\mu_k, \mbf{x}_{-k}, \mbf{x}_{-k}+\sgn(i-k)\tfrac{\pi}{2}, -\mbf{x}_{-k}, -\mbf{x}_{-k}+\tfrac{\pi}{2}] \\
		&= [\mu_k, \mbf{x}_{-k}, \mbf{x}_{-k}+\mbf{h}\tfrac{\pi}{2}, -\mbf{x}_{-k}, -\mbf{x}_{-k}+\tfrac{\pi}{2}],
	\end{align*}
	with, for example, $\bs{\alpha}_{\cdot,k}$ denoting all $\alpha$ parameters involving index $k$ and $\mbf{h}_j = -1$ if $j < k$ and $\mbf{h}_j = 1$ otherwise. 
	Then since $\int f(\mbf{x}; \bs{\phi}) \, dx_k = 2 \pi I_0(A)$ we find that the conditional density is von Mises with concentration $A$ and mean $\Delta$. 
\end{proof}

%% file: texsupp/posnegdepend.tex
The dependence between two circular variables can be measured using a correlation coefficient, $\rho_c$, analogous to the Pearson correlation coefficient for linear analysis \cite{jammalamadaka1988correlation},
\begin{align*}
	\rho_c = \frac{E\{\sin(X_i-\mu_i)\sin(X_j-\mu_j)\}}{\sqrt{\text{Var}(\sin(X_i-\mu_i))\text{Var}(\sin(X_j-\mu_j))}}
\end{align*}
where $\mu$ represents a mean circular direction.
Using variance properties and trigonometric identities we have
\begin{align}
	\rho_c = \frac{E\{\cos(X_i-X_j-(\mu_i-\mu_j))-\cos(X_i+X_j-(\mu_i+\mu_j))\}}{2\sqrt{E\{\sin^2(X_i-\mu_i)\}E\{\sin^2(X_j-\mu_j)\}}} .\label{eq:rhoc}
\end{align}
The first component of the numerator measures the positive correlations from the concentration of $X_i-X_j-(\mu_i-\mu_j)$ and the second component measures the  negative (or \textit{reflectional}) correlations from the concentration of $X_i-(-X_j)-(\mu_i-(-\mu_j))$.
Analogous to the real-valued data, it is important to note that both positive and negative circular correlations are possible and both are needed to fully define circular dependence between two variables.

The numerator of Equation \ref{eq:rhoc} can be rewritten as
\begin{align*}
	E\left\{
\begin{bmatrix}
         \cos(X_i-X_j) \\ \sin(X_i-X_j)
\end{bmatrix}^T 
\begin{bmatrix}
          \alpha \\ \beta
\end{bmatrix} -
\begin{bmatrix}
          \cos(X_i+X_j) \\ \sin(X_i+X_j)
\end{bmatrix}^T 
\begin{bmatrix}
           \gamma \\ \delta
\end{bmatrix}
\right\}  
\end{align*}
where $\alpha=\cos(\mu_i-\mu_j)$, $\beta=\sin(\mu_i-\mu_j)$, $\gamma=\cos(\mu_i+\mu_j)$, $\delta=\sin(\mu_i+\mu_j)$.
This shows that the dependence between angles may be decomposed into a four-term linear combination involving the sines and cosines of the phase differences and phase sums, where the phase difference terms correspond to the positive correlation and the phase sum terms correspond to the negative correlation. 
This corresponds to the phase sum and phase difference terms that appear in the torus graph density, reinforcing the interpretation of the different pairwise coupling parameters as reflecting positive and negative rotational dependence.
To illustrate the distinction between positive and negative rotational dependence, we show bivariate torus graphs with positive, negative, or both kinds of dependence in Figure \ref{fig:bivar_posneg}, and show what trial-to-trial rotational and reflectional covariance would look like in Figure \ref{fig:phasesum}.
For the case of uniform marginal distributions, the circular correlation coefficient becomes \cite{JammSengupta}: 
\begin{align*}
	\rho_c = \frac{R_{Xi-Xj}-R_{Xi+Xj}}{2\sqrt{E\{\sin^2(X_i-\mu_i)\}E\{\sin^2(X_j-\mu_j)\}}}
\end{align*}
where $R_{Xi-Xj} \equiv | E\{\exp(\textbf{i}(X_i-X_j)) \}|$ corresponds to the positive correlation and $R_{Xi+Xj} \equiv | E\{\exp(\textbf{i}(X_i+X_j)) \}|$ corresponds to the negative correlation.
The theoretical Phase Locking Value (PLV), for which an estimator is given in Equation \ref{eq:plv} of the main text, is equal to $R_{Xi-Xj}$. 
This shows that PLV is similar to a measure of positive circular correlation under the assumption of uniform marginal distributions (when the denominator of the circular correlation coefficient would be equal to 1).

%% file: texsupp/figures.tex
\begin{figure}[h]
	\centering
	\includegraphics[width=0.9\linewidth]{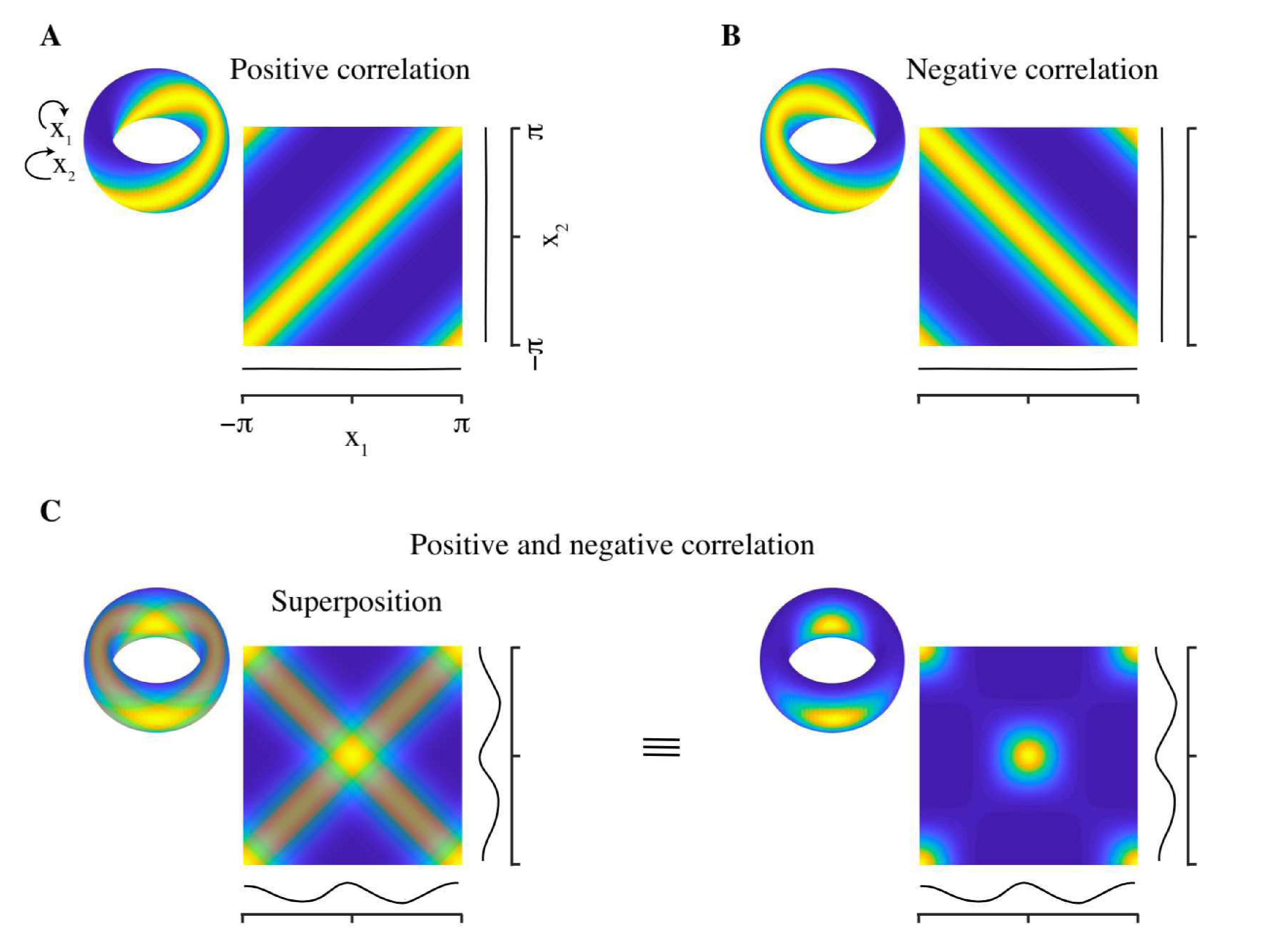}
	\caption{Bivariate torus graph densities with uniform marginal distributions  shown on the torus and flattened on $[-\pi,\pi]$ under positive, negative, or both kinds of circular covariance. (A) Coupling parameters chosen to induce only positive correlation (coupling based on phase differences).  (B) Coupling parameters chosen to induce only negative correlation (coupling based on phase sums). (C) Equal amounts of positive and negative coupling result in a distribution with two isotropic modes; the superposition figure (left) is shown to provide intuition about the resulting distribution (right).}
	\label{fig:bivar_posneg}
\end{figure}

\begin{figure}[h]
	\centering
	\includegraphics[width=1.0\linewidth]{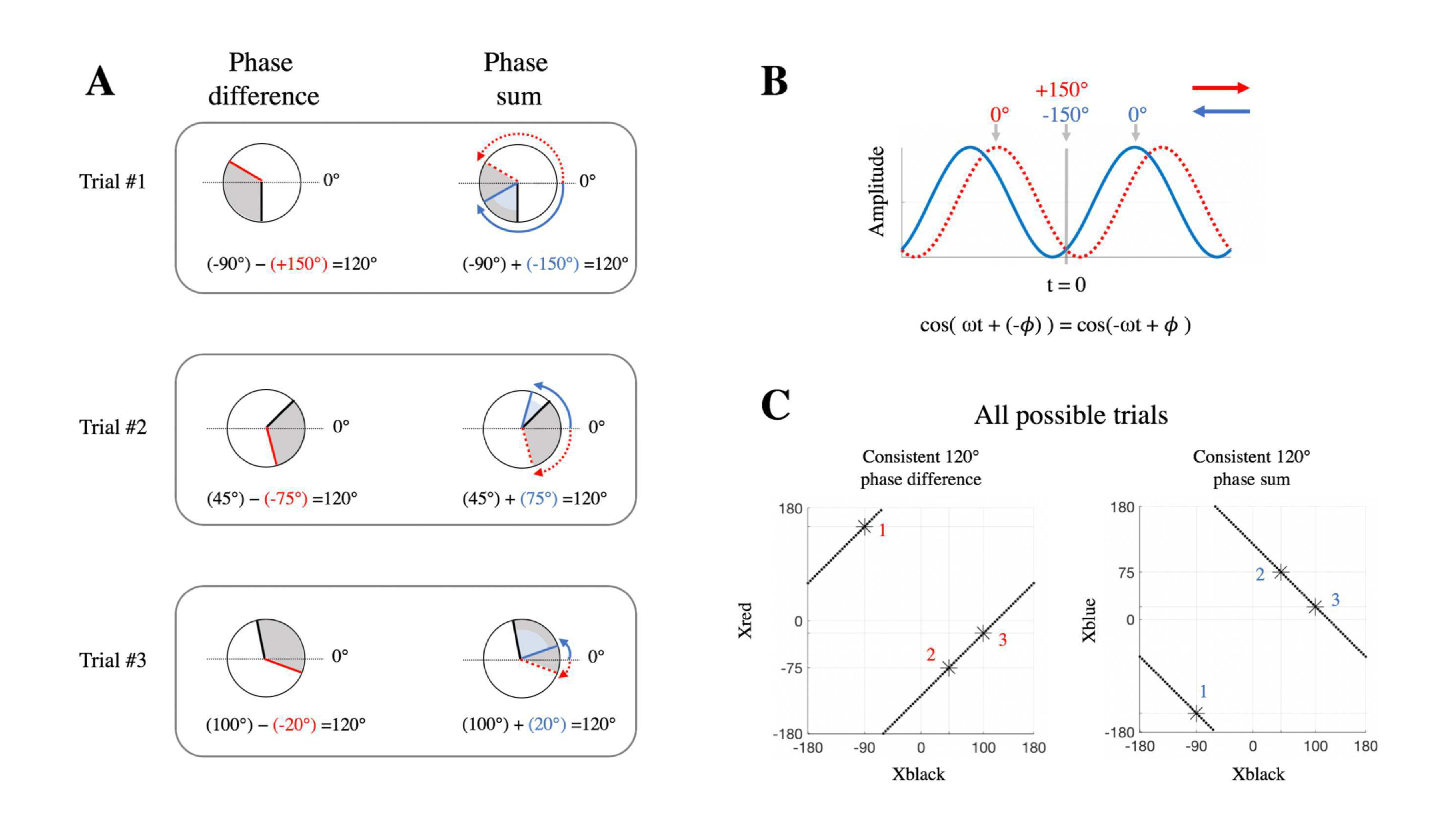}
	\caption{
	Intuition about rotational and reflectional dependence. A) Illustration of $120^\circ$ rotational (positive) and $120^\circ$ reflectional (negative) dependence across three hypothetical observations. 
    Rotational dependence implies a consistent phase offset between oscillations across trials (shaded gray angles) while reflectional dependence implies a consistent phase sum (shaded blue angles) that corresponds to a consistent phase offset between oscillations after one of the oscillations has been reflected with respect to $0^\circ$. 
    B) Example of an oscillation (blue) and its reflection with respect to $0^\circ$ (dashed red). 
    The reflected signal is leading in time, i.e. $\phi= 150^\circ$, whereas the blue signal is lagging in time, i.e. $\phi=-150^\circ$. 
    This demonstrates that we can think of the reflected signal as moving across time in the opposite direction. 
    In neural data this phenomenon could arise, for example, if there was bidirectional communication. 
    C) The lines indicate dependence between phase angles used in panel A, with rotational dependence on the left and reflectional dependence on the right and trials shown in panel A marked with stars. 
    When phases have uniform marginal distributions across trials, but exhibit phase coupling, positive dependence is observed in the bivariate relationship on the left as a line with fixed orientation at $45^\circ$; rotation produces a shift along the anti-diagonal. 
    On the other hand, negative dependence is observed in the bivariate relationship on the right as a line with fixed orientation at $135^\circ$; reflection produces a shift along the diagonal. 
    In the current example, the respective shifts are at $120^\circ$; see Figure \ref{fig:bivar_posneg} to compare with the case of $0^\circ$. 
    Weaker dependence blurs the relationship line but does not change the orientation. 
}
	\label{fig:phasesum}
\end{figure}

\begin{figure}
    \centering
    \includegraphics[width=\linewidth]{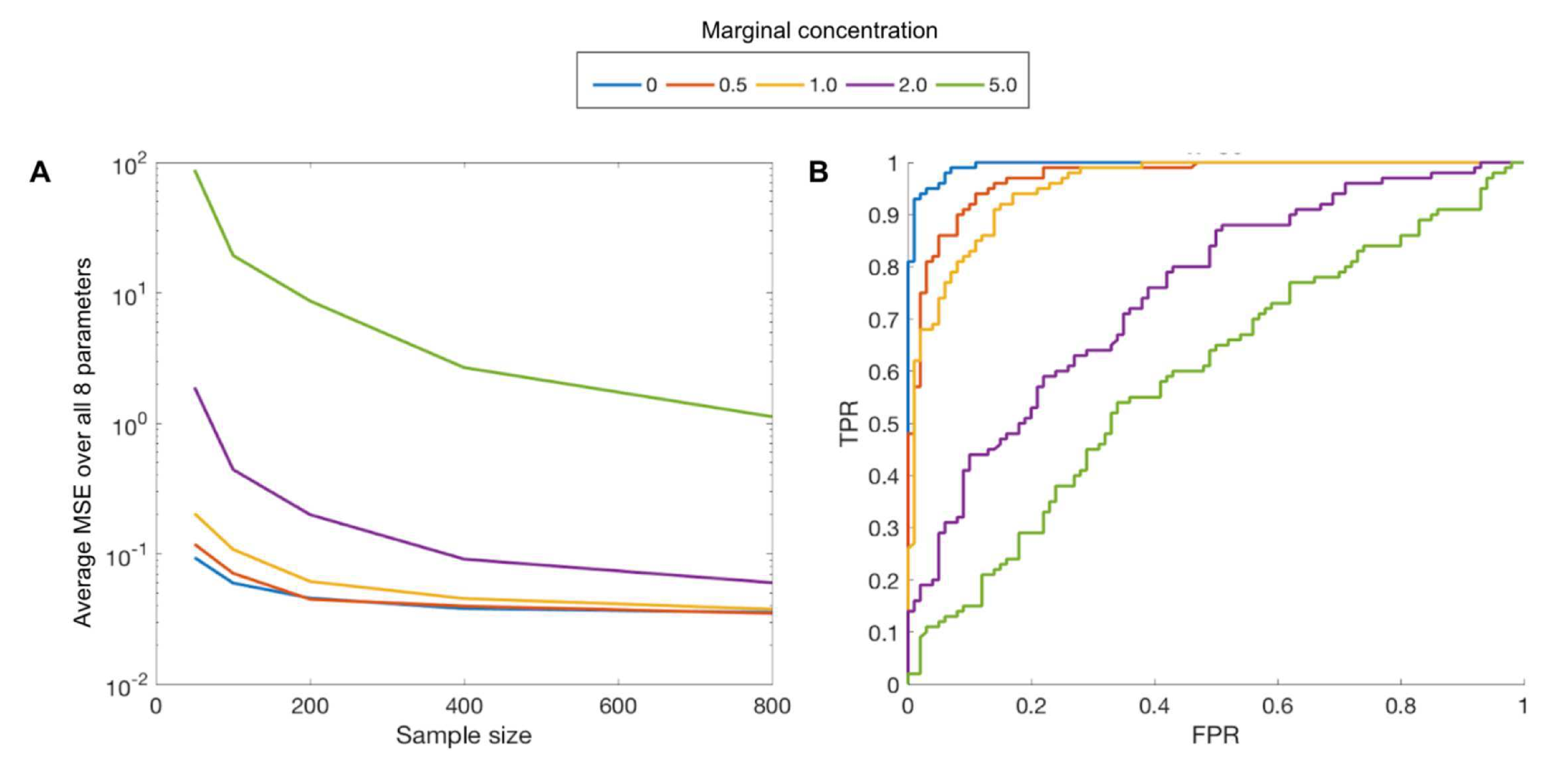}
    \caption{
    In data simulated from a bivariate torus graph, the average MSE over all parameters is shown in panel (A) as a function of sample size and marginal concentration. The MSE is higher overall when marginal concentration is high. (B) ROC curves averaged over 200 simulations, half of which had no edge and half of which had an edge between the variables, as a function of marginal concentration for a fixed sample size (N=50), suggesting that structure recovery is also diminished when high marginal concentration is present.
    }
    \label{fig:bivkappROCMSE}
\end{figure}

\begin{figure}
    \centering
    \includegraphics[width=0.9\linewidth]{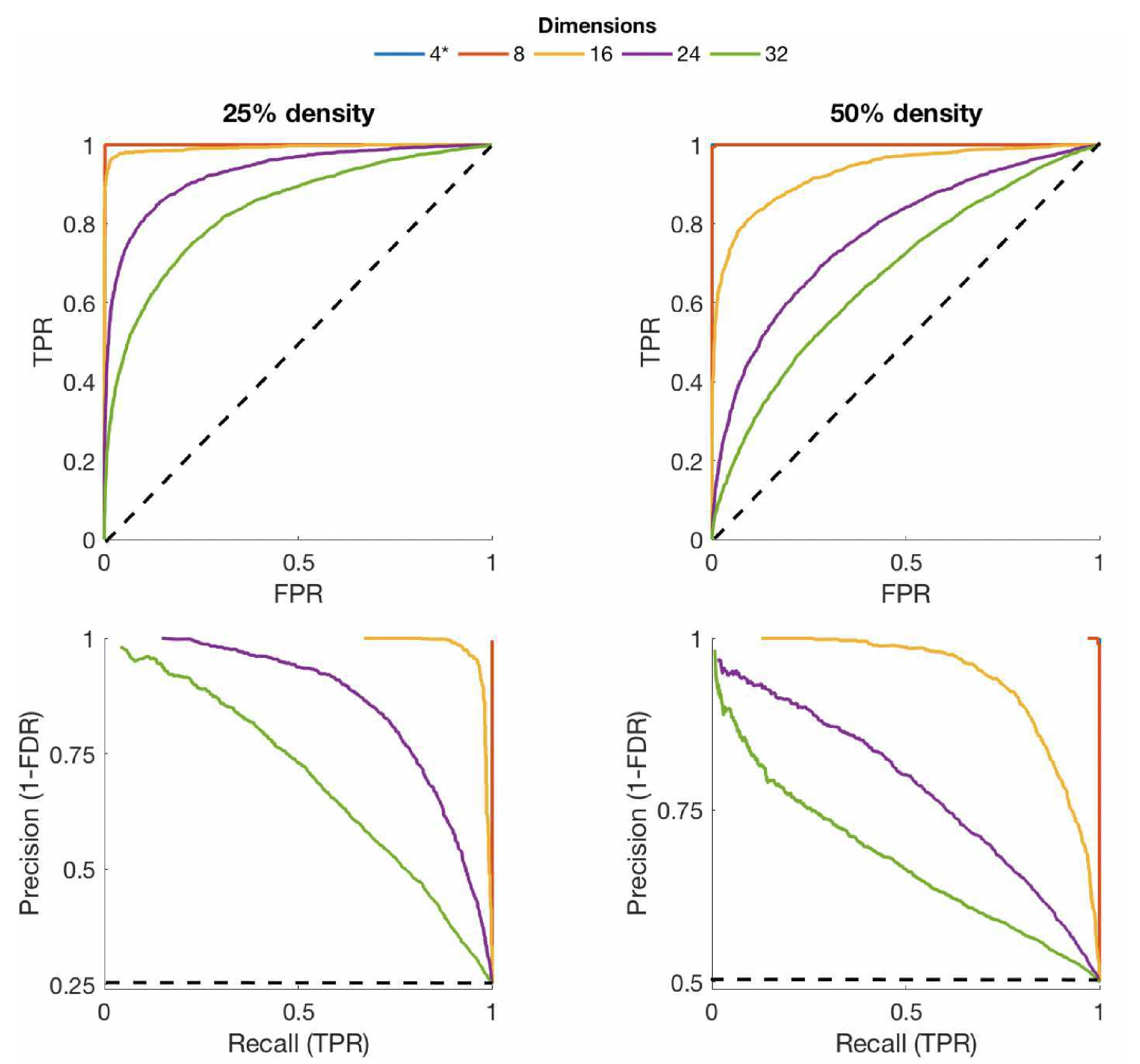}
    \caption{
    Further detail on the simulation results shown in Figure \ref{fig:AUCbydensity} for a sample size of 840 (matching the real LFP data). 
    Top row: ROC curves colored by dimension for two different underlying edge densities (averaged across 30 simulations). 
    Bottom row: averaged precision curves corresponding to the same densities as the top row.
    }
    \label{fig:ROCprecbydim}
\end{figure}

\begin{figure}
    \centering
    \includegraphics[width=0.7\linewidth]{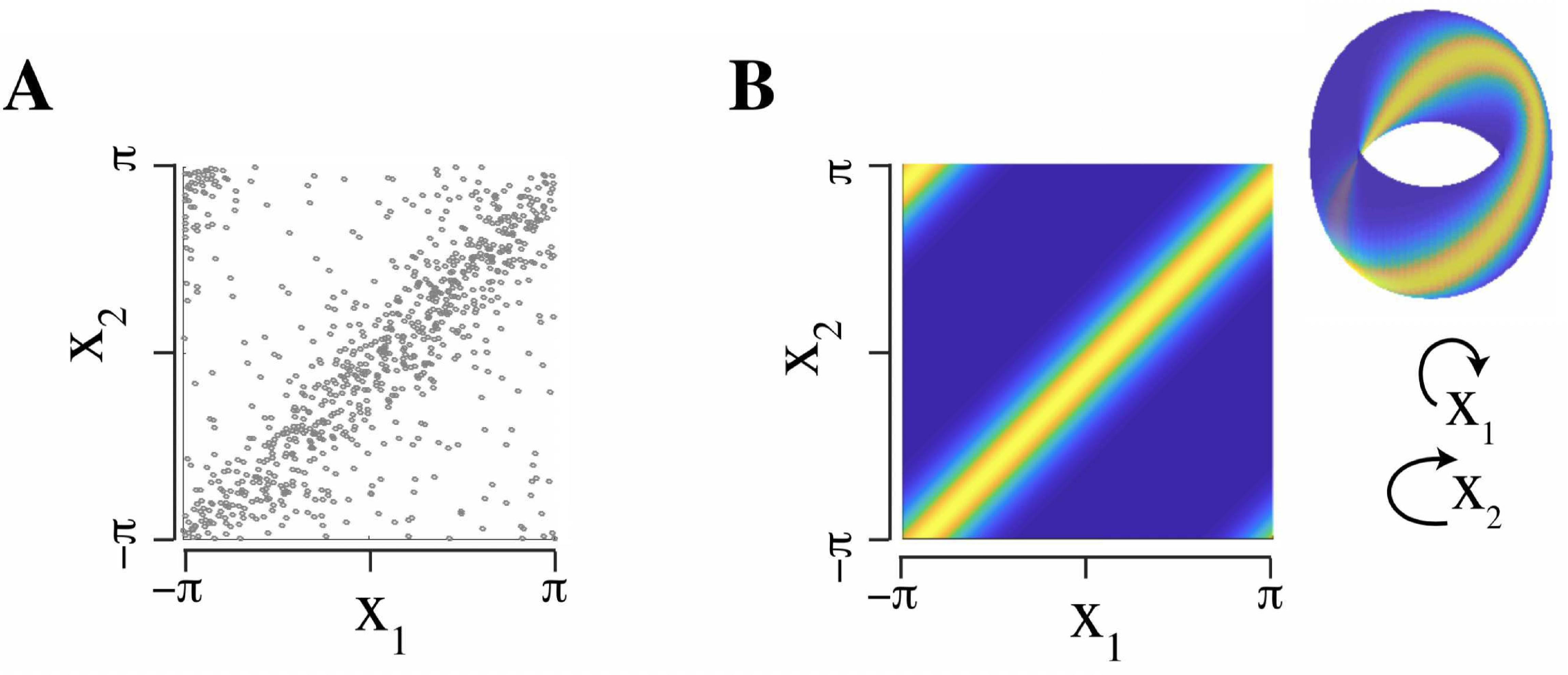}
    \caption{
    (A) Scatter plot of data from two channels in dentate gyrus (DG). The angles follow a pattern of positive dependence similar to the simulated data of Figure \ref{FigRepeatedTrials_part1}, which was used to demonstrate the need for circular wrapping when modeling dependent phase angles. (B) Fitted torus graph density on the plane and torus.
    }
    \label{fig:bivscatterDG}
\end{figure}

\clearpage

\begin{figure}
    \centering
    \includegraphics[width=0.8\linewidth]{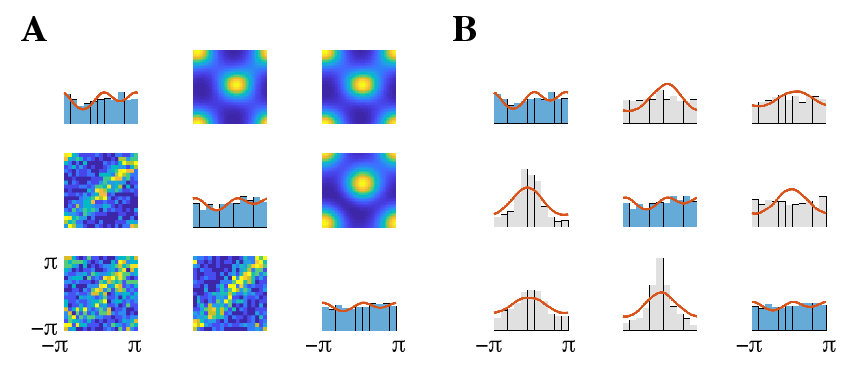}
    \caption{
    Similar to Figure \ref{fig:eda}, but using the sine model as the theoretical distribution. 
    The sine model fails to accurately fit this data set, which is evident in the bivariate dependence and sufficient statistics. 
    (A) Along the diagonal are the marginal distributions of the three phase angles, where the real data is represented by blue histograms and the theoretical marginal densities from the sine model are overlaid as solid red traces. 
    Two-dimensional distributions (off-diagonal) show bivariate relationships, with theoretical densities above the diagonal and real data represented using two-dimensional histograms below the diagonal. 
    The multimodal behavior of the sine model is apparent in the two-dimensional distributions, which do not appear to match the real data.
    (B) Plots along the diagonal same as panel A. 
    Below the diagonal are distributions of pairwise phase differences and above the diagonal are distributions of pairwise phase sums, represented by histograms for the real data and by solid red density plots for the theoretical torus graph model. 
    In contrast to the torus graph model, the sine model fails to accurately capture the distributions of the sufficient statistics from these data.
    }
    \label{fig:edawithsinemodel}
\end{figure}

\begin{figure}[h]
	\centering
	\includegraphics[width=\linewidth]{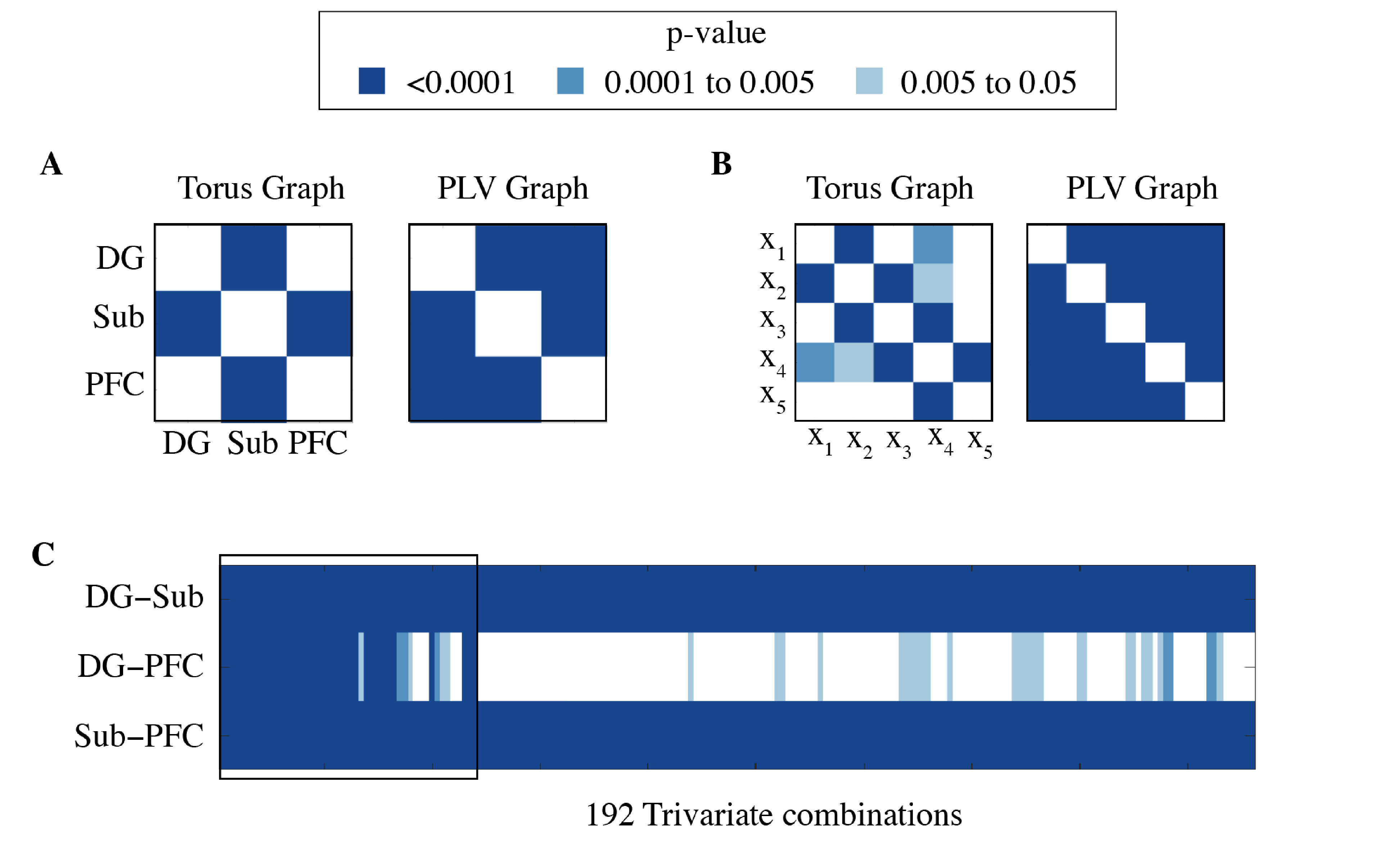}
	\caption{(A) Adjacency matrices for the three-dimensional LFP analysis with entries colored by $p$-value. PLV $p$-values are very small for all connections, while torus graph $p$-values reflect finer structure, such as PFC-Sub coupling that is apparently more salient than PFC-DG coupling. The adjacency matrix for the trivariate network is a representative combination of channels reflecting the pattern that dominates in all trivariate combinations. (B) Same as (A) but for the five-dimensional LFP analysis, where the torus graph reveals nearest-neighbor structure along the linear probe in CA3 that PLV misses. (C) Edgewise $p$-values for each edge in each possible trivariate graph (composed of each combination of electrodes from each of the three areas). Note that channels 6 and 7 (boxed region) of DG are on the border between DG and CA3 and may be picking up signals from CA3; omitting these channels gives stronger evidence of an overall lack of connections between DG and PFC (corresponding to the adjacency matrix in (A)).}
	\label{fig:three_five_var}
\end{figure}

\begin{figure}[h]
	\centering
	\includegraphics[width=0.9\linewidth]{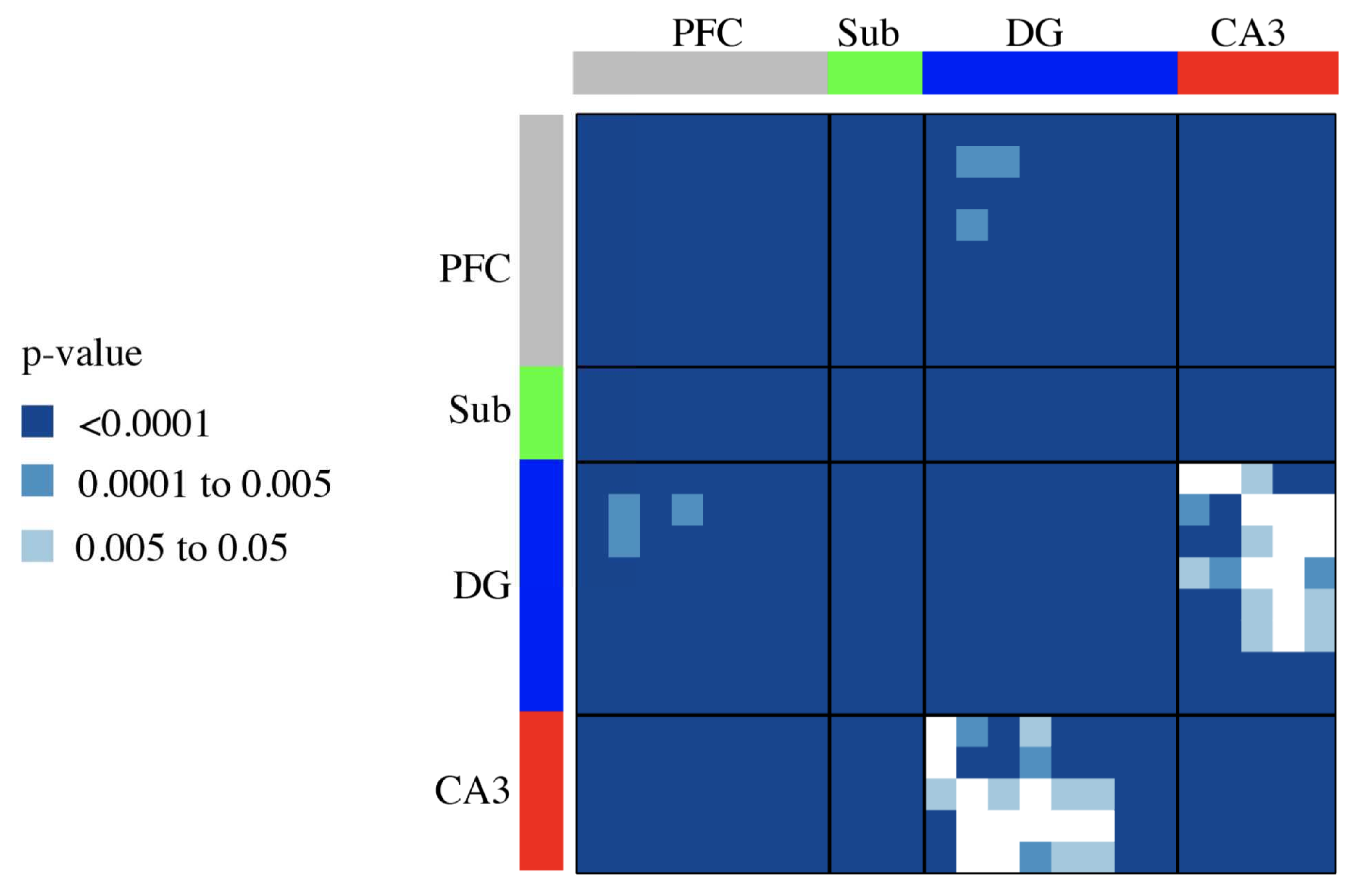}
	\caption{Adjacency matrix for PLV graph with edgewise $p$-values determined using Rayleigh's test of uniformity on the circle for each pairwise phase difference. Entries are colored by $p$-value and, compared to the torus graph adjacency matrix Figure \ref{fig:AcrossROI}C, there is very little noticeable structure in the graph even for very small $p$-value thresholds (aside from a lack of edges between CA3 and DG).}
	\label{fig:plvadj}
\end{figure}

\begin{figure}
    \centering
    \includegraphics[width=\linewidth]{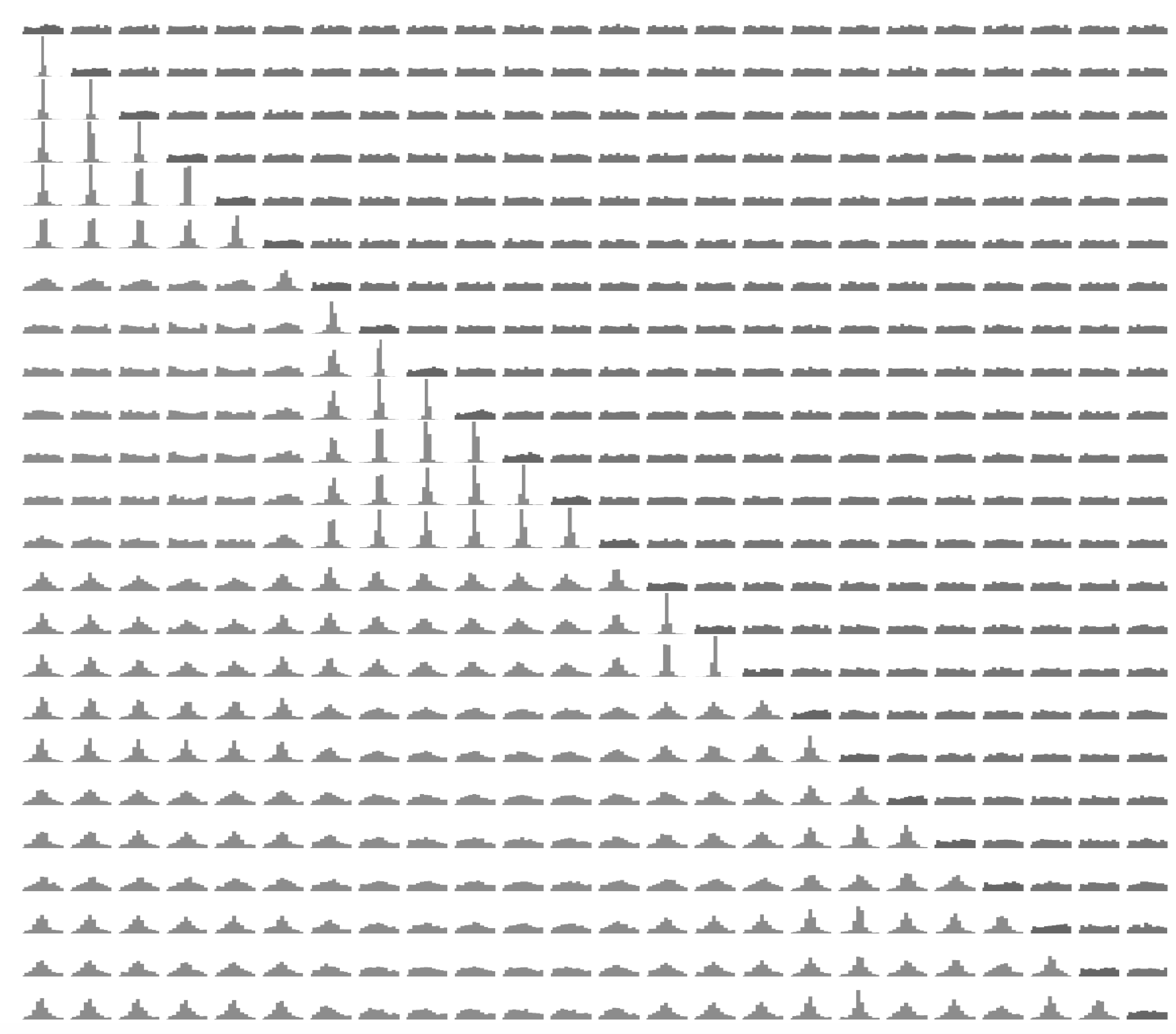}
    \caption{For the 24-dimensional real LFP data, the diagonal shows univariate histograms which appear to have low concentration in all cases. Below the diagonal are histograms of phase differences between pairs of angles, showing some highly concentrated distributions suggesting rotational dependence; above the diagonal are histograms of phase sums, showing very little concentration, suggesting there is not strong evidence for reflectional dependence in these data.}
    \label{fig:suppsuffstateda}
\end{figure}

\begin{figure}[h]
	\centering
	\includegraphics[width=\linewidth]{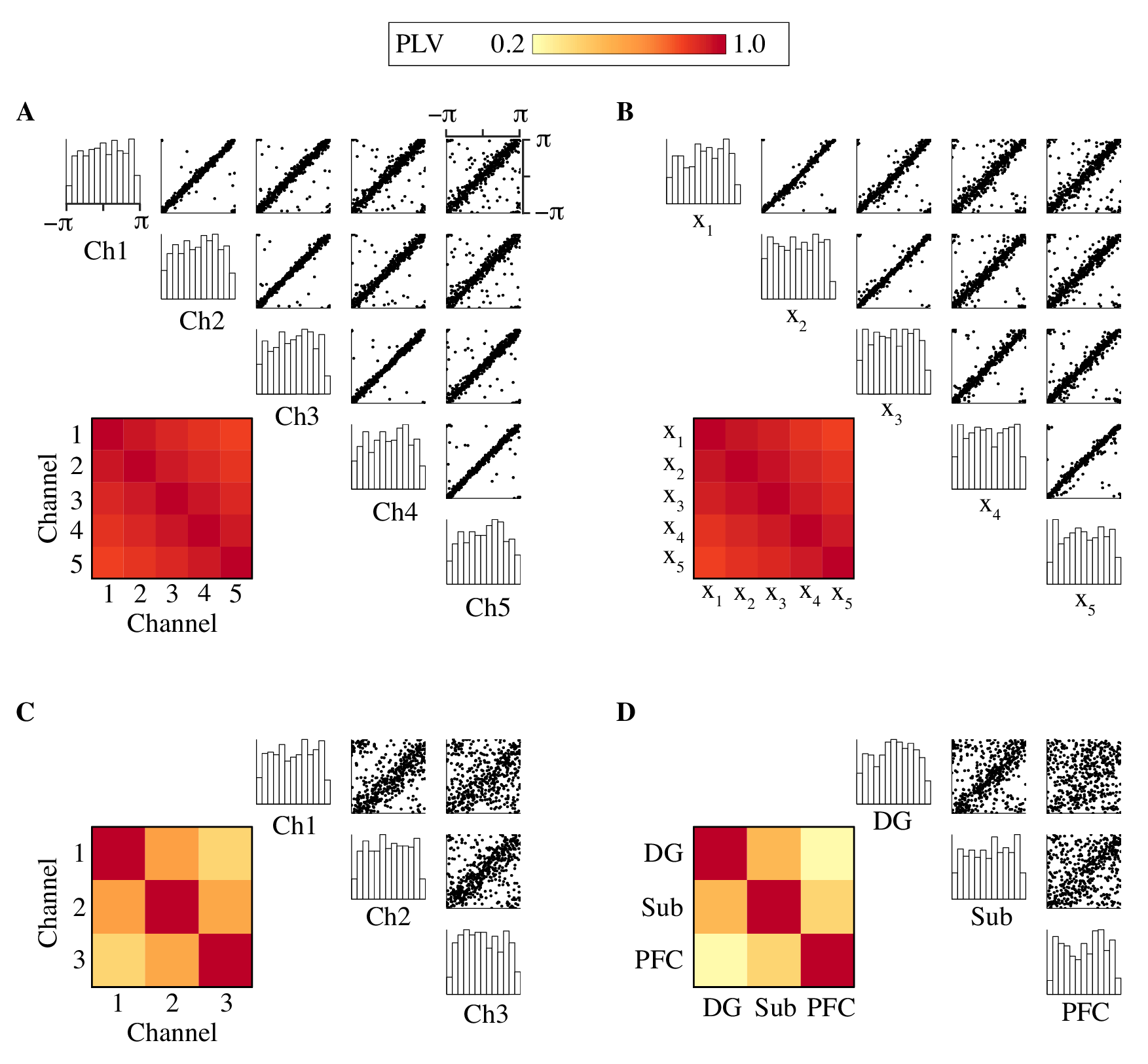}
	\caption{Examples of simulated and real data to demonstrate the validity of the simulation process. Upper right: histograms and pairwise scatter plots, bottom left: estimated PLV matrices (color scale 0.2 to 1, with red indicating higher PLV values). (A) Simulated 5-channel data with linear probe structure. (B) Real 5-channel data from CA3. (C) Simulated 3-channel data. (D) Real 3-channel data from separate regions (DG, Sub, and PFC). }
	\label{fig:toyvalidation}
\end{figure}

\begin{figure}[h]
	\centering
	\includegraphics[width=\linewidth]{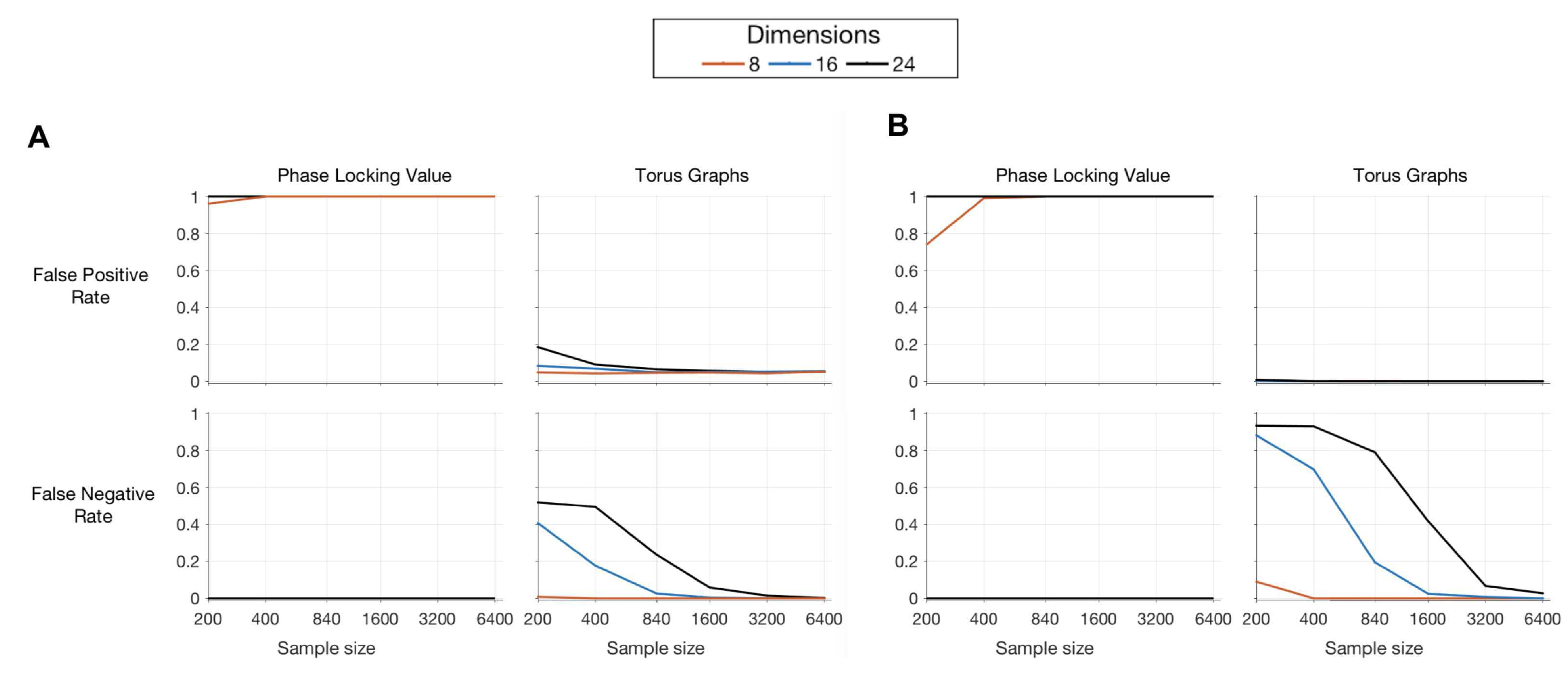}
	\caption{Investigation of False Positive Rate (FPR) and False Negative Rate (FNR) for graphs of varying dimensions as sample size increases. (A) FPR and FNR for PLV (left) and torus graphs (right) using an alpha level of 0.05 for the edgewise hypothesis tests with no Bonferroni correction for the number of edges. PLV has high FPR for all sample sizes while torus graphs control the FPR; on the other hand, PLV has low FNR, but torus graphs is more conservative and for low sample sizes may be missing some edges. (B) Same as A, but with Bonferroni correction.}
	\label{fig:fnrfpr}
\end{figure}